\documentclass[11pt,reqno]{amsart}

\usepackage{amsmath}
\usepackage{amssymb}
\usepackage{a4wide}
\usepackage{bbm}
\usepackage{graphics}
\usepackage{color}
\usepackage[applemac]{inputenc}
\usepackage[sans]{dsfont}

\newtheorem{theorem}{Theorem}[section]
\newtheorem{lemma}[theorem]{Lemma}
\newtheorem{proposition}[theorem]{Proposition}

\numberwithin{equation}{section}

\renewcommand{\d}{\mathrm{d}}
\newcommand{\el}{\mathrm{el}}
\newcommand{\n}{\mathrm{n}}
\newcommand{\ph}{\mathrm{ph}}
\renewcommand{\i}{\mathrm{i}}
\newcommand{\DETAILS}[1]{}

\begin{document}

\title[Hyperfine splitting in non-relativistic QED]{Hyperfine splitting of the dressed hydrogen atom ground state in non-relativistic QED}

\author{L. Amour}
\address[L. Amour]{Laboratoire de Math\'ematiques EDPPM \\ FRE-CNRS 3111, Universit\'e de Reims \\ Moulin de la Housse - BP 1039, 51687 REIMS Cedex 2, France}
\email{laurent.amour@univ-reims.fr}
\author{J. Faupin}
\address[J. Faupin]{Institut de Math{\'e}matiques de Bordeaux \\
UMR-CNRS 5251, Universit{\'e} de Bordeaux 1 \\
351 cours de la lib{\'e}ration, 33405 Talence Cedex, France}
\email{jeremy.faupin@math.u-bordeaux1.fr}

\date{\today}

\begin{abstract}
We consider a spin-$\frac{1}{2}$ electron and a spin-$\frac{1}{2}$ nucleus interacting with the quantized electromagnetic field in the standard model of non-relativistic QED. For a fixed total momentum sufficiently small, we study the multiplicity of the ground state of the reduced Hamiltonian. We prove that the coupling between the spins of the charged particles and the electromagnetic field splits the degeneracy of the ground state.
\end{abstract}

\maketitle

\section{Introduction}
This paper is concerned with the spectral analysis of the quantum Hamiltonian associated with a free hydrogen atom, in the context of non-relativistic QED. Before describing our result more precisely, let us neglect for a while the corrections due to quantum electrodynamics. In the following we recall a few well-known properties of the spectrum of Hydrogen. For more details, we refer the reader to classical textbooks on Quantum Mechanics (see, e.g., \cite{Me,CTDL}). See also \cite{BS}.

Consider a neutral hydrogenoid system composed of one electron and one nucleus. The system is supposed to be free, in the sense that no external potential acts on it. The two charged particles (the electron and the nucleus) are supposed to be non-relativistic. In particular, relativistic corrections will not be taken into account in the discussion below.

Under sufficiently strong approximations, the spectrum of the hydrogen atom Hamiltonian is explicitly known. Assume, indeed, that the nucleus is infinitely heavy and that the two charged particles are spinless. Then the Schr{\"o}dinger operator in $\mathrm{L}^2( \mathbb{R}^3 )$ associated with the system reads
\begin{equation}\label{eq:h^sch}
\frac{p^2_\el}{2m_{\el}} - \frac{ \alpha }{ | x_\el | }.
\end{equation}
Here the units are chosen such that $\hbar = c = 1$, where $\hbar = h / 2\pi$, $h$ is the Planck constant, and $c$ is the velocity of light. The mass of the electron is denoted by $m_\el$, $\alpha = e^2$ is the fine-structure constant (with $e$ the charge of the electron), and $x_\el$, respectively $p_\el = - \i \nabla_{x_\el}$, represents the position of the electron, respectively its momentum.

The spectrum of \eqref{eq:h^sch} consists of an infinite increasing sequence of negative isolated eigenvalues $(e_j)_{j\ge0}$, and the branch of absolutely continuous spectrum $[0,\infty)$. The ground state is non-degenerate, that is the eigenvalue $e_0$ is simple, while the excited eigenvalues $(e_j)_{j \ge 1}$ are degenerate (the degeneracy of the $(j+1)^{\mathrm{th}}$ eigenvalue, $e_j$, being equal to $(j+1)^2$).

Introducing the degrees of freedom associated with the spin of the electron, the hydrogen atom can be described by the following Hamiltonian acting on $\mathrm{L}^2( \mathbb{R}^3 ; \mathbb{C}^2 )$:
\begin{equation}\label{eq:h^so}
\frac{p^2_\el}{2m_{\el}} - \frac{ \alpha }{ | x_\el | } + \frac{1}{ 4 m_\el^2 } \frac{ \alpha }{ | x_\el |^3 } \sigma^\el \cdot ( x_\el \wedge p_\el ),
\end{equation}
where $\sigma^\el = ( \sigma^\el_1, \sigma^\el_2, \sigma^\el_3 )$ are the Pauli matrices for the electron spin. The last term in the preceding expression is the \emph{spin-orbit interaction}. It can be derived from the Dirac equation in the non-relativistic regime. Together with other relativistic corrections that have been neglected in \eqref{eq:h^so}, the spin-orbit interaction is responsible for the \emph{fine structure} of the spectrum of Hydrogen. In particular, the ground state of \eqref{eq:h^so} is twice-degenerate, while the spin-orbit coupling splits the degeneracy of the excited eigenvalues $(e_j)_{j\ge 1}$. This can be justified by means of standard perturbation theory, the unperturbed Hamiltonian being given by the expression \eqref{eq:h^sch} seen as an operator on $\mathrm{L}^2( \mathbb{R}^3 ; \mathbb{C}^2 )$.

Assume now that the nucleus is a fixed particle with a spin equal to $\frac{1}{2}$ and a finite mass. In order to take the effect of the spin nucleus into account, one can study the following Pauli Hamiltonian in $\mathrm{L}^2( \mathbb{R}^3 ; \mathbb{C}^4 )$:
\begin{equation}\label{eq:h^pauli}
\frac{1}{2m_{\el}} ( p_\el - \alpha^{ \frac{1}{2} } A_{\n}( x_\el ) )^2 - \frac{ \alpha }{ | x_\el | } - \frac{ \alpha^{ \frac{1}{2} } }{ 2 m_\el } \sigma^{\el} \cdot B_{ \n }( x_\el ) + \frac{1}{ 4 m_\el^2 } \frac{ \alpha }{ | x_\el |^3 } \sigma^\el \cdot ( x_\el \wedge p_\el ).
\end{equation}
Here $A_\n( x_\el )$ is the vector potential of the electromagnetic field generated by the nucleus, and $B_\n( x_\el ) = \i p_\el \wedge A_\n( x_\el )$. Notice that $A_\n( x_\el )$ can be expressed in terms of the Pauli matrices $\sigma^\n = ( \sigma_1^\n , \sigma_2^\n , \sigma_3^\n )$ associated with the spin of the nucleus as $A_\n (  x_\el ) = \mathrm{c} \alpha^{ 1/2 } ( \sigma^\n \wedge x_\el ) / ( m_\n | x_\el |^3 )$, where $m_\n$ is the mass of the nucleus and $\mathrm{c}$ is a positive constant. The expression \eqref{eq:h^pauli} includes the spin-orbit interaction.

The Hamiltonian \eqref{eq:h^pauli} allows one to justify the so-called \emph{hyperfine structure} of the spectrum of Hydrogen. Again, the argument follows from perturbation theory. The unperturbed part is still given by \eqref{eq:h^sch}, now seen as an operator on $\mathrm{L}^2( \mathbb{R}^3 ; \mathbb{C}^4 )$. Hence the degeneracy of the unperturbed ground state eigenvalue, $e_0$, is equal to $4$. Under the influence of the perturbation terms appearing in \eqref{eq:h^pauli} (more precisely, under the influence of the term $\sigma^{\el} \cdot B_{ \n }( x_\el )$), $e_0$ splits into two parts: a simple eigenvalue associated with a unique ground state, and a $3$-fold degenerate eigenvalue. Let us mention that this splitting of the ground state explains the famous \emph{$21$-cm Hydrogen line}. Besides, a similar hyperfine splitting of the excited eigenvalues $(e_j)_{j \ge 1}$ occurs.

Of course, in order to get a refined picture of the spectrum of Hydrogen, one should also treat the nucleus as a moving quantum particle. The corresponding physical system is then translation invariant, and one is led to study the relative Hamiltonian in the center of mass frame. In other words, one can consider the Hamiltonian obtained by putting the total momentum equal to $0$. For instance, in the simplest case where both the electron and the nucleus are spinless, the relative Hamiltonian is given by the expression \eqref{eq:h^sch}, except that $x_\el$ and $p_\el$ are replaced by the relative position and the relative momentum respectively, and $m_\el$ is replaced by the reduced mass $\mu = m_\el m_\n / ( m_\el + m_\n )$ (with $m_\n$ the mass of the nucleus).

In this paper, we investigate the hyperfine structure of the hydrogen atom Hamiltonian in non-relativistic QED. Let us mention that non-relativistic QED provides a suitable framework to rigorously justify radiative decay and Bohr's frequency condition. In particular, save for the ground state, all stationary states are expected to turn into metastable states with a finite lifetime. Since the unperturbed eigenvalues are embedded into the essential spectrum, usual perturbation theory does not apply, and proving the latter statement involves highly non-trivial analysis (see \cite{BFS1,BFS2,AFFS,Sigal} for the case of atomic systems with an infinitely heavy nucleus). The present work focuses on the question of the nature of the ground state, for the model of a freely moving hydrogen atom at a fixed total momentum. The difference between the unperturbed and the perturbed ground state energies is referred to as \emph{the Lamb shift}. Our main concern is to study the \emph{degeneracy} of the ground state. More precisely, we aim at establishing that a hyperfine splitting of the ground state does occur in the framework of non-relativstic QED.

Let us now describe our main result in more details.

We consider a non-relativistic hydrogen atom interacting with the quantized electromagnetic field in the standard model of non-relativistic QED. Both the electron and the nucleus are treated as moving particles, so that the total Hamiltonian, $H_g$, is translation invariant. Here $g$ denotes a coupling parameter depending on the fine-structure constant, $\alpha$, related to the ``strength'' of the interaction between the charged particles and the electromagnetic field (precise definitions will be given in Subsection \ref{subsection:definition} below). The translation invariance implies that $H_g$ admits a direct integral decomposition, $H_g \sim \int_{ \mathbb{R}^3 } H_g( P ) \d P$, with respect to the total momentum $P$ of the system. This paper is devoted to the study of the degeneracy of the ground state energy of the dressed hydrogen atom at a fixed total momentum, $E_g(P) := \inf \sigma ( H_g(P) )$.

In \cite{AGG2}, it is established that, for $g$ and $P$ sufficiently small, $E_g(P)$ is an eigenvalue of $H_g(P)$, that is $H_g(P)$ has a ground state. We also mention \cite{LMS1} where the existence of a ground state for $H_g(P)$ is obtained for any value of $g$, under the assumption that $E_g(0) \le E_g(P)$. Using a method due to \cite{Hir2}, it is proven in \cite{AGG2} that the multiplicity of $E_g(P)$ cannot exceed the multiplicity of $E_0(P) := \inf \sigma ( H_0(P) )$, where $H_0(P) := H_{g=0}(P)$ denotes the non-interacting Hamiltonian. In other words,
\begin{equation}\label{eq:multiplicity_intro}
( 0 < ) \dim \, \mathrm{Ker} \, ( H_g(P) - E_g(P) ) \le \dim \, \mathrm{Ker} \, ( H_0(P) - E_0(P) ).
\end{equation}
Our purpose is to determine whether the inequality in \eqref{eq:multiplicity_intro} is strict, or, on the contrary, is an equality. 

Of course, in the same way as in the case where the coupling to the quantized electromagnetic field is neglected (see the discussion at the beginning of this introduction), the multiplicity of $E_g(P)$ depends on the value of the spins of the charged particles. If the spin of the electron is neglected and the spin of the nucleus is equal to 0, then $E_0(P)$ is simple, and hence, according to \eqref{eq:multiplicity_intro}, $E_g(P)$ is also a simple eigenvalue. In particular, \eqref{eq:multiplicity_intro} is an equality.

If the spin of the electron is taken into account, and the spin of the nucleus is equal to 0, then $E_0(P)$ is twice-degenerate. Using Kramer's degeneracy theorem (see \cite{LMS2}), one can prove that the multiplicity of $E_g(P)$ is even. Therefore, by \eqref{eq:multiplicity_intro}, $E_g(P)$ is also twice-degenerate, and hence \eqref{eq:multiplicity_intro} is again an equality. We refer the reader to \cite{HS,Spohn,Sa,Hir,LMS2} for results on the twice-degeneracy of the ground state of various QED models.

Consider now a hydrogen atom composed of a spin-$\frac{1}{2}$ electron and a spin-$\frac{1}{2}$ nucleus (e.g. a proton). In this case, the multiplicity of $E_0(P)$ is equal to 4. Our main result states that
\begin{equation}\label{eq:splitting_intro}
\dim \, \mathrm{Ker} \, ( H_g(P) - E_g(P) ) < \dim \, \mathrm{Ker} \, ( H_0(P) - E_0(P) ) = 4,
\end{equation}
for $g \neq 0$ small enough. Equation \eqref{eq:splitting_intro} can be interpreted as a hyperfine splitting of the ground state of $H_g(P)$. In other words, the Hamiltonian of a freely moving hydrogen atom at a fixed total momentum in non-relativistic QED contains hyperfine interaction terms which split the degeneracy of the ground state, in the same way as for the Pauli Hamiltonian of Quantum Mechanics mentioned above. Pursuing the analogy with the Pauli Hamiltonian \eqref{eq:h^pauli}, one can conjecture that $E_g(P)$ is simple. Generally speaking, a way to establish the uniqueness of the ground state of a given Hamiltonian $H$ consists in showing that $e^{-t H}$ is positivity improving for all $t > 0$. In several cases, this can be done by constructing a functional integral representation for the semi-group $e^{-t H }$ (see e.g. \cite{Simon} for Schr{\"o}dinger operators under various general assumptions). To our knowledge, however, such a functional representation for the model of two spin-$\frac{1}{2}$ particles minimally coupled to the quantized radiation field does not presently exist (we refer to the recent work \cite{HL} for the case of one spin-$\frac{1}{2}$ particle minimally coupled to the radiation field). Here, we shall follow a different approach; Proving the simplicity of $E_g(P)$ is beyond the scope of this paper.

In addition, in relation with the 21-cm hydrogen line mentioned above, one can expect that a resonance appears near the ground state energy $E_g(P)$, with a very small imaginary part. Showing this would presumably require the use of complex dilatations together with renormalization techniques as in \cite{BFS1}.

The case of a nucleus of spin $\ge1$ is not considered here (for instance, the nucleus of deuterium, composed of one proton and one neutron, can be treated as a spin-$1$ particle), but we expect that a similar hyperfine splitting of the ground state occurs in this case also.

As for positively charged hydrogenoid ions, the question of the existence of a ground state is more subtle than for the hydrogen atom. Indeed, it is proven in \cite{HH} that the Hamiltonian of a positive ion at a fixed total momentum in non-relativistic QED does not have a ground state in Fock space. This result should be compared with the corresponding one for the model of a freely moving, dressed non-relativistic electron in non-relativistic QED, which has been studied recently by several authors (see, among other papers, \cite{Chen,CF,BCFS,HH,CFP,LMS2,FP}). For the latter model, it is established that a ground state exists in a non-Fock representation (see \cite{CF}).

The same results (the absence of a ground state in Fock space, and the existence of a ground state in a non-Fock representation) are proven in \cite{AFGG1,AFGG2}, for a dressed non-relativistic electron which interacts with a classical magnetic field, following an approach different from \cite{CF}. We expect that the method of \cite{AFGG2} can be adapted to prove the existence of a ground state for a renormalized Hamiltonian (in a non-Fock representation) associated with the dressed helium ion $\mathrm{He}^+$, at a fixed total momentum. Since the nucleus of $\mathrm{He}^+$, composed of two protons, can be treated as a spin-0 particle, one can conjecture that the ground state is twice-degenerate according to Kramer's theorem (see \cite{LMS2}).

Regarding the model of a hydrogenoid ion whose nucleus has a non-zero spin, we do not know whether it is possible to adapt \cite{AFGG2}. Indeed, an important ingredient in \cite{AFGG2} lies in the regularity of the ground state energy with respect to the total momentum. For the model considered in \cite{AFGG2}, the latter property can be established using the method developed in \cite{Pizzo,CFP,FP}. However, if a hyperfine splitting of the ground state occurs, it is not known, to our knowledge, how to study the regularity of $E_g(P)$ with respect to $P$.

Let us finally mention that the ground state degeneracy of the non-relativistic hydrogen atom confined by its center of mass (see \cite{AF,Fa}) could also be analyzed by the techniques developed here, provided that both the electron and the nucleus have a spin equal to $\frac{1}{2}$.

\section{Definition of the model and statement of the main result}

Before stating our main theorem, let us precisely define the model under consideration.

\subsection{Definition of the model}\label{subsection:definition}

We consider a mobile, non-relativistic hydrogen atom, interacting with the quantized electromagnetic field in the Coulomb gauge. In the standard model of non-relativistic QED, the Hamiltonian associated with this system acts on the Hilbert space
\begin{equation}
\mathcal{H} := \mathcal{H}_{\mathrm{at}} \otimes \mathcal{H}_{\mathrm{ph}},
\end{equation}
where
\begin{equation}
\mathcal{H}_{\mathrm{at}} := \mathrm{L}^2( \mathbb{R}^3 ; \mathbb{C}^2 ) \otimes \mathrm{L}^2( \mathbb{R}^3 ; \mathbb{C}^2 ) \sim \mathrm{L}^2( \mathbb{R}^6 ; \mathbb{C}^4 )
\end{equation}
is the Hilbert space for the charged particles (the electron and the nucleus), and
\begin{equation}
\mathcal{H}_{\mathrm{ph}} := \mathbb{C} \oplus \bigoplus_{n=1}^\infty S_n \left [ \mathrm{L}^2( \mathbb{R}^3 \times \{1,2\} )^{ \otimes^n} \right ]
\end{equation}
is the symmetric Fock space for the photons. Here $S_n$ denotes the symmetrization operator.

The units are chosen such that the Planck constant $\hbar = h / 2 \pi$ and the velocity of light $c$ are equal to 1. The Hamiltonian of the system, $H^{ \mathrm{SM} }$, is formally given by the expression
\begin{align}
H^{\mathrm{SM} }  := & \frac{1}{2m_\el} \left( p_\el - \alpha^{\frac{1}{2}} A( x_\el )\right )^2 + \frac{1}{2m_\n} \left( p_\n + \alpha^{\frac{1}{2}} A( x_\n ) \right)^2 + V(x_\el,x_\n) + H_{\rm ph} \notag \\
& -\frac{ \alpha^{\frac{1}{2}} }{2m_\el} \sigma^\el \cdot B(x_\el) + \frac{ \alpha^{\frac{1}{2}} }{2m_\n} \sigma^\n \cdot B(x_\n), \label{eq:defH}
\end{align}
where $x_\el$ (respectively $x_\n$) denotes the position of the electron (respectively the position of the nucleus), and $p_\el := - \i \nabla_{x_\el}$ (respectively $p_\n := - \i \nabla_{x_\n}$) is the momentum operator of the electron (respectively of the nucleus). The parameter $\alpha$ denotes the fine-structure constant. For $x \in \mathbb{R}^3$, the vectors $A(x)$ and $B(x)$ are defined by
\begin{align}
& A(x) := \frac{1}{2\pi} \sum_{\lambda =1,2} \int_{ \mathbb{R}^3 } \frac{ \chi_\Lambda (k) }{ |k|^{\frac{1}{2}} } \varepsilon^\lambda(k) \left [ e^{ - \i k \cdot x } a^*_{\lambda }(k) + e^{ \i k \cdot x}a_{\lambda} (k)\right ] \d k, \label{eq:defA} \\
& B(x) := -\frac{\i}{2\pi} \sum_{\lambda =1,2} \int_{ \mathbb{R}^3 } |k|^{\frac{1}{2} } \chi_\Lambda(k) \big( \frac{k}{|k|} \wedge \varepsilon^{\lambda}(k) \big) \left [ e^{ - \i k \cdot x} a^*_{\lambda} (k) - e^{ \i k \cdot x}a_{\lambda} (k)\right ] \d k, \label{eq:defB}
\end{align}
where the polarization vectors are chosen in the following way:
\begin{equation}
\varepsilon^1(k): = \frac{  ( k_2 , -k_1 , 0 ) }{ \sqrt{ k_1^2 + k_2^2 } } \quad , \quad \varepsilon^2(k) := \frac{ k }{ |k| } \wedge \varepsilon^1 (k) = \frac{ ( - k_1 k_3 , - k_2 k_3 , k_1^2 + k_2^2 ) }{ \sqrt{ k_1^2 + k_2^2 } \sqrt{ k_1^2 + k_2^2 + k_3^2 } }.
\end{equation}
In \eqref{eq:defA} and \eqref{eq:defB}, $\chi_\Lambda(k)$ denotes an ultraviolet cutoff function which, for the sake of concreteness, we choose as
\begin{equation}\label{eq:chi_Lambda}
\chi_\Lambda(k) := \mathds{1}_{ |k| \le \Lambda \alpha^2 }(k).
\end{equation}
Here, $\Lambda$ is supposed to be a given arbitrary (large and) positive parameter. As explained in \cite{BFS2,Sigal}, the model is physically relevant if we assume that $1 \ll \Lambda \ll \alpha^{-2}$. The reason for introducing $\alpha^2$ into the definition \eqref{eq:chi_Lambda} will appear below (see \eqref{eq:tilde_chi_Lambda}).

For any $h \in \mathrm{L}^2( \mathbb{R}^3 \times \{1,2\} )$, we set
\begin{equation}
a^*(h) := \sum_{\lambda=1,2} \int_{\mathbb{R}^3} h( k,\lambda ) a^*_\lambda(k) \d k, \quad a(h) := \sum_{\lambda=1,2} \int_{\mathbb{R}^3} \bar h( k,\lambda ) a_\lambda( k) \d k,
\end{equation}
and
\begin{equation}
\Phi(h) := a^*(h) + a(h),
\end{equation}
where the usual creation and annihilation operators, $a^*_\lambda(k)$ and $a_\lambda(k)$, obey the canonical commutation relations
\begin{equation}
[ a_\lambda (k) , a_{\lambda'}(k') ] = [ a^*_\lambda(k) , a^*_{\lambda'}(k') ] = 0 , \quad [ a_\lambda ( k ) , a^*_{\lambda'}( k' ) ] = \delta_{\lambda\lambda'} \delta ( k - k' ).
\end{equation}
Hence, in particular, for $j \in \{ 1,2,3 \}$, we have
\begin{equation}
A_j(x) = \Phi( h^A_j(x) ), \quad \text{and} \quad B_j(x) = \Phi( h^B_j(x) ), 
\end{equation}
with
\begin{align}
& h^A_j(x,k,\lambda) := \frac{1}{2\pi} \frac{ \chi_\Lambda(k) }{ |k|^{ \frac{1}{2} } } \varepsilon^\lambda_j(k) e^{ - \i k \cdot x }, \label{eq:h^A} \\
& h^B_j(x,k,\lambda) := - \frac{\i}{2\pi} |k|^{\frac{1}{2}} \chi_\Lambda(k) \left ( \frac{ k }{ |k| } \wedge \varepsilon^\lambda(k) \right )_j e^{ - \i k \cdot x }. \label{eq:h^B}
\end{align}
The Coulomb potential $V(x_\el,x_\n)$ is given by
\begin{equation}
V(x_\el,x_\n) \equiv V ( x_\el - x_\n ) := - \frac{ \alpha }{ | x_\el - x_\n | },
\end{equation}
and $H_\ph$ is the Hamiltonian of the free photon field, defined by
\begin{equation}
H_{\mathrm{ph}} := \sum_{ \lambda = 1, 2 } \int_{ \mathbb{R}^3 } |k| a^*_\lambda(k) a_\lambda(k) \d k.
\end{equation}
The $3$-uples $\sigma^\el = ( \sigma^\el_1 , \sigma^\el_2 , \sigma^\el_3 )$ and $\sigma^\n = ( \sigma^\n_1 , \sigma^\n_2 , \sigma^\n_3 )$ are the Pauli matrices associated with the spins of the electron and the nucleus respectively. They can be written as $4 \times 4$ matrices in the following way:
\begin{footnotesize}
\begin{align}
\sigma^\el_1 &=
  \left (
  \begin{array}{cccc}
   0 & 0 & 1 & 0 \\
   0 & 0 & 0 & 1 \\
   1 & 0 & 0 & 0 \\
   0 & 1 & 0 & 0
  \end{array}
  \right ) , \quad
  \sigma^\el_2 =
    \left (
  \begin{array}{cccc}
   0 & 0 & - \i & 0 \\
   0 & 0 & 0 & -\i \\
   \i & 0 & 0 & 0 \\
   0 & \i & 0 & 0
  \end{array}
  \right ) , \quad
  \sigma^\el_3 =
    \left (
  \begin{array}{cccc}
   1 & 0 & 0 & 0 \\
   0 & 1 & 0 & 0 \\
   0 & 0 & -1 & 0 \\
   0 & 0  & 0 & -1
  \end{array}
  \right ) ,  \label{eq:sigma_el} \\
 \sigma^\n_1 &=
  \left (
  \begin{array}{cccc}
   0 & 1 & 0 & 0 \\
   1 & 0 & 0 & 0 \\
   0 & 0 & 0 & 1 \\
   0 & 0 & 1 & 0
  \end{array}
  \right ) , \quad
  \sigma^\n_2 =
    \left (
  \begin{array}{cccc}
   0 & -\i & 0 & 0 \\
   \i & 0 & 0 & 0 \\
   0 & 0 & 0 & -\i \\
   0 & 0 & \i & 0
  \end{array}
  \right ) , \quad
  \sigma^\n_3 =
    \left (
  \begin{array}{cccc}
   1 & 0 & 0 & 0 \\
   0 & -1 & 0 & 0 \\
   0 & 0 & 1 & 0 \\
   0 & 0  & 0 & -1
  \end{array}
  \right ). \label{eq:sigma_n}
\end{align}
\end{footnotesize}

In order to exhibit the perturbative behavior of the interaction between the charged particles and the photon field, we proceed to a change of units. More precisely, let $\mathcal{U} : \mathcal{H} \to \mathcal{H}$ be the unitary operator associated with the scaling 
\begin{align}
(x_\el , x_\n , k_1 , \lambda_1 , \dots , k_n , \lambda_n ) \mapsto ( x_\el / \alpha , x_\n / \alpha , \alpha^2 k_1 , \lambda_1 , \dots , \alpha^2 k_n , \lambda_n ).
\end{align}
We have
\begin{align}\label{eq:defH_2}
& \frac{ 1 }{ \alpha^2 } \mathcal{U} H^{\mathrm{SM}} \mathcal{U}^* = \frac{1}{2m_\mathrm{el}} \left( p_\mathrm{el} - \alpha^{\frac{3}{2}} \tilde A( \alpha x_\mathrm{el} )\right )^2 + \frac{1}{2m_\mathrm{n}} \left( p_\mathrm{n} + \alpha^{\frac{3}{2}} \tilde A( \alpha x_\mathrm{n} ) \right)^2 \notag \\
& - \frac{ 1 }{ | x_\mathrm{el} - x_\mathrm{n} | } + H_{\rm ph} -\frac{ \alpha^{\frac{3}{2}} }{2m_\mathrm{el}} \sigma^{\mathrm{el}} \cdot \tilde B( \alpha x_\mathrm{el} ) + \frac{ \alpha^{\frac{3}{2}} }{2m_\mathrm{n}} \sigma^{\mathrm{n}} \cdot \tilde B( \alpha x_\mathrm{n} ),
\end{align}
where $\tilde A$ and $\tilde B$ are defined in the same way as $A$ and $B$, except that the ultraviolet cutoff function $\chi_\Lambda(k)$ is replaced by 
\begin{equation}\label{eq:tilde_chi_Lambda}
\tilde \chi_\Lambda(k) := \chi_\Lambda( \alpha^2 k ) = \mathds{1}_{ |k| \le \Lambda }(k).
\end{equation}
To simplify the notations, we redefine $\tilde \chi_\Lambda = \chi_\Lambda$, $A = \tilde A$ and $B = \tilde B$. Setting $g := \alpha^{\frac{3}{2}}$, we are thus led to study the Hamiltonian
\begin{align}\label{eq:defH_g}
H_g^{\mathrm{SM}} :=& \frac{1}{2m_\mathrm{el}} \left( p_\mathrm{el} - g A( g^{\frac{2}{3}} x_\mathrm{el} )\right )^2 + \frac{1}{2m_\mathrm{n}} \left( p_\mathrm{n} + g A( g^{\frac{2}{3}} x_\mathrm{n} ) \right)^2 \notag \\
& - \frac{ 1 }{ | x_\mathrm{el} - x_\mathrm{n} | } + H_{\rm ph} -\frac{ g }{2m_\mathrm{el}} \sigma^{\mathrm{el}} \cdot B( g^{\frac{2}{3}} x_\mathrm{el} ) + \frac{ g }{2m_\mathrm{n}} \sigma^\mathrm{n} \cdot B( g^{\frac{2}{3}} x_\mathrm{n} ).
\end{align}

Let the total mass, $M$, and the reduced mass, $\mu$, be defined respectively by
\begin{equation}
M := m_\el + m_\n , \quad \frac{ 1 }{ \mu } := \frac{ 1 }{ m_\el } + \frac{ 1 }{ m_\n }.
\end{equation}
Let
\begin{align}
& r := x_\el - x_\n , \quad R := \frac{ m_\el }{ M } x_\el + \frac{ m_\n }{ M } x_\n , \\
& \frac{ p_r }{ \mu } := \frac{ p_\el }{ m_\el } - \frac{ p_\n }{ m_\n } , \quad P_R := p_\el + p_\n.
\end{align}
For $g=0$, the Hamiltonian $H_0^{\mathrm{SM}} := H_{g=0}^{\mathrm{SM}}$ is given by
\begin{equation}
H_0^{\mathrm{SM}} = \frac{ p_\el^2 }{ 2m_\el } + \frac{ p_\n^2 }{ 2m_\n } - \frac{1}{ | x_\el - x_\n | } + H_{\mathrm{ph}} = H_R + H_r + H_{\mathrm{ph}},
\end{equation}
where the Schrödinger operators $H_R$ and $H_r$ on $\mathrm{L}^2( \mathbb{R}^3 )$ are defined by
\begin{equation}\label{eq:H_RandH_r}
H_R := \frac{ P_R^2 }{ 2M }, \quad H_r := \frac{ p_r^2 }{ 2 \mu } - \frac{1}{|r|}.
\end{equation}
Let us note that the spectrum of $H_R$ consists of the branch of essential spectrum $[0,\infty)$, whereas the spectrum of $H_r$ is composed of an increasing sequence of isolated eigenvalues $(e_0,e_1,e_2,\dots)$ accumulating at 0, and the essential spectrum $[ 0 , \infty )$. The first eigenvalue of $H_r$ is
\begin{equation}
e_0 = - \frac{ \mu }{ 2 },
\end{equation}
and a normalized eigenstate associated with $e_0$ is given by
\begin{equation}\label{eq:phi0}
\phi_0(r) := ( \pi^{-1} \mu^3 )^{\frac{1}{2}} e^{ - \mu |r| }.
\end{equation}

To conclude this subsection, we recall the definition of the photon number operator, $\mathcal{N}_\ph$, which will be used in the sequel:
\begin{equation}
\mathcal{N}_\ph := \sum_{\lambda=1,2} \int_{\mathbb{R}^3} a^*_\lambda(k) a_\lambda(k) \d k.
\end{equation}

\subsection{Fiber decomposition}

The Hamiltonian $H_g^{\mathrm{SM}}$ is translation invariant in the sense that $H_g^{ \mathrm{SM} }$ formally commutes with the total momentum operator $P_{ \mathrm{tot} } := P_R + P_\ph$, where $P_\ph$ denotes the momentum operator of the photon field, given by the expression
\begin{equation}
P_\ph := \sum_{\lambda = 1,2} \int_{ \mathbb{R}^3 } k a^*_\lambda(k) a_\lambda(k) \d k.
\end{equation}
In the same way as in \cite{AGG2}, it follows that $H_g^{\mathrm{SM}}$ can be decomposed into a direct integral, which is expressed in the following proposition.
\begin{proposition}[\cite{AGG2}]\label{prop:decomposition}
There exists $g_c>0$ such that for all $|g| \le g_c$, the following holds: the Hamiltonian $H_g^{\mathrm{SM}}$ given by the formal expression \eqref{eq:defH_g} identifies with a self-adjoint operator which is unitarily equivalent to the direct integral $\int_{\mathbb{R}^3}^\oplus H_g(P) \d P$,
\begin{equation}\label{eq:unitary_equivalence}
H_g^{\mathrm{SM}} \sim \int_{ \mathbb{R}^3 }^\oplus H_g ( P ) \d P.
\end{equation}
For all $P \in \mathbb{R}^3$, $H_g(P)$ is a self-adjoint operator acting on the Hilbert space
\begin{equation}
\mathcal{H}(P) := \mathrm{L}^2( \mathbb{R}^3 ; \mathbb{C}^4 ) \otimes \mathcal{H}_{\mathrm{ph}} \sim \mathbb{C}^4 \otimes \mathrm{L}^2( \mathbb{R}^3 , \d r ) \otimes \mathcal{H}_{\mathrm{ph}},
\end{equation}
with domain $D(H_g(P)) = D( H_0(P) )$, and $H_g(P)$ is given by the expression:
\begin{align}
 H_g( P ) =& \frac{1}{2m_\el} \left( \frac{ m_\el }{ M } ( P - P_\ph ) + p_r - g A( \frac{ m_\el }{ M } g^{\frac{2}{3}} r )\right )^2 \notag \\
 & + \frac{1}{2m_\n} \left( \frac{ m_\n }{ M } ( P - P_\ph ) - p_r + g A( - \frac{ m_\n }{M} g^{\frac{2}{3}} r ) \right)^2 \notag \\
& - \frac{ 1 }{ |r| } + H_{\rm ph}   -\frac{g}{2m_\el} \sigma^{\el} \cdot B( \frac{ m_\el }{ M } g^{\frac{2}{3}} r )
   + \frac{g}{2m_\n} \sigma^{\n} \cdot B ( - \frac{ m_\n }{ M }  g^{\frac{2}{3}} r ).
\end{align}
\end{proposition}
Let us mention that the direct integral decomposition \eqref{eq:unitary_equivalence} remains true for an arbitrary value of the coupling constant $g$ (see \cite{LMS1}). However, in this paper, we shall only be interested in the small coupling regime.

For $g=0$, the fiber Hamiltonian $H_0(P) := H_{g=0}(P)$ reduces to the diagonal operator
\begin{align}
H_0(P) = H_r + \frac{1}{2M} ( P - P_\ph )^2 + H_{\mathrm{ph}},  \label{eq:def_H0}
\end{align}
where $H_r$ is the Schrödinger operator defined in \eqref{eq:H_RandH_r}. Let $\Omega$ denote the photon vacuum in $\mathcal{H}_{\mathrm{ph}}$. One can verify that 
\begin{equation}
E_0(P) := \inf \sigma ( H_0(P) ) = e_0 + \frac{P^2}{2M},
\end{equation}
and that $e_0 + P^2/2M$ is an eigenvalue of multiplicity 4 of $H_0(P)$. Moreover, the associated normalized eigenstates can be written under the form $y \otimes \phi_0 \otimes \Omega$, where $y$ is an arbitrary normalized element in $\mathbb{C}^4$.

The operator $H_0(P)$ is treated as an unperturbed Hamiltonian, the perturbation $W_g(P) := H_g(P) - H_0(P)$ being given by
\begin{align}
W_g(P) =& - \frac{ g }{ m_\el } \left ( \big ( \frac{ m_\el }{ M } ( P - P_\ph ) + p_r \big ) \cdot A ( \frac{ m_\el }{ M } g^{\frac{2}{3}} r ) \right ) \notag \\ 
& + \frac{ g }{ m_\n } \left ( \big ( \frac{ m_\n }{ M } ( P - P_\ph ) - p_r \big ) \cdot A ( - \frac{ m_\n }{ M } g^{\frac{2}{3}} r ) \right ) \notag \\ 
& + \frac{ g^2 }{ 2m_\el } A ( \frac{ m_\el }{ M } g^{\frac{2}{3}} r )^2 + \frac{ g^2 }{ 2m_\n } A ( - \frac{ m_\n }{ M } g^{\frac{2}{3}} r )^2 \notag \\
& -\frac{g}{2m_\el} \sigma^{\el} \cdot B( \frac{ m_\el }{ M } g^{\frac{2}{3}} r ) + \frac{g}{2m_\n} \sigma^{\n} \cdot B ( - \frac{ m_\n }{ M } g^{\frac{2}{3}} r ). \label{eq:Wg}
\end{align}
Note that, due to the choice of the Coulomb gauge, the operators $A ( m_\el g^{2 / 3} r / M )$ and $A ( - m_\n g^{ 2 / 3 } r / M )$ commute both with $p_r$ and $P_\ph$.

As mentioned in the introduction, this paper is concerned with the nature of the bottom of the spectrum of the perturbed Hamiltonian, $E_g(P) := \inf \sigma ( H_g(P) )$. In other words, we would like to determine the behavior of the unperturbed eigenvalue $E_0(P)$ under the perturbation $W_g(P)$. We emphasize that, since $E_0(P)$ is embedded into the continuous spectrum of $H_0(P)$, usual perturbation theory of isolated eigenvalues of finite multiplicity does not apply. In \cite{AGG2}, it is proven that, for $g$ and $P$ sufficiently small, $H_g(P)$ has a ground state of multiplicity $\le 4$, that is $E_g(P)$ is an (embedded) eigenvalue of multiplicity $\le 4$ of $H_g(P)$. The aim of our work is to study more precisely the multiplicity of $E_g(P)$, still assuming that $g$ and $P$ are sufficiently small.

\subsection{Main result}

Our main result is stated in the following theorem.
\begin{theorem}\label{thm:main}
There exists $g_c>0$ and $p_c>0$ such that, for any $0 < |g| \le g_c$ and $0 \le |P| \le p_c$,
\begin{equation}\label{eq:main}
\mathrm{dim} \, \mathrm{Ker} \, ( H_g(P) - E_g(P) ) < 4.
\end{equation}
\end{theorem}
This theorem shows that the interaction between the charged particles with spins and the photon field splits the degeneracy of the ground state. As explained in the introduction, this splitting of the ground state can be seen as a manifestation of the hyperfine interaction between the electron and the nucleus. It may be expected that a renormalization group analysis adapted from \cite{BFS1,BCFS}, together with a careful study of the second order term in the expansion of $E_g(P)$ w.r.t. $g$, would lead to a more precise result. A feature of the method developed in the present work is its brevity.

Our proof of Theorem \ref{thm:main} is based on a contradiction argument and the use of the \emph{Feshbach-Schur identity}. Let us sketch the argument more precisely. For technical convenience, we shall work with the Hamiltonian obtained from $H_g(P)$ by Wick ordering (see Section \ref{section:Feshbach}). Let $P_0$ denote the projection onto the eigenspace associated with the eigenvalue $E_0(P)$ of $H_0(P)$, and $\bar P_0 := \mathds{1} - P_0$. Note that $P_0 = \mathds{1} \otimes P_{\phi_0} \otimes P_\Omega$ in the tensor product $\mathbb{C}^4 \otimes \mathrm{L}^2( \mathbb{R}^3 ) \otimes \mathcal{H}_\ph$, where $P_{\phi_0}$ denotes the projection onto the eigenspace associated with the ground state $\phi_0$ of $H_r$, and $P_\Omega$ denotes the projection onto the Fock vacuum. We also set $P_\rho := \mathds{1} \otimes P_{\phi_0} \otimes \mathds{1}_{ H_f \le \rho }$, and $\bar P_\rho := \mathds{1} - P_\rho$, where $\rho$ is a suitably chosen positive parameter (depending on $g$ in such a way that $g^2 \ll \rho \ll 1$, see Sections \ref{section:Feshbach} and \ref{section:proof}). By the Feshbach-Schur identity, we will see that, for all $\varepsilon > 0$,
\begin{equation}\label{eq:Feshbach_identity}
P_\rho \big [ H_g(P) - E_g(P) + \varepsilon \big ]^{-1} P_\rho = F_\rho(\varepsilon)^{-1},
\end{equation}
where $F_\rho(\varepsilon)$ denotes the \emph{Feshbach-Schur operator}:
\begin{align}
F_\rho(\varepsilon) =& \big ( H_0(P) - E_g(P) + \varepsilon \big ) P_\rho + P_\rho W_g(P) P_\rho \notag \\
& - P_\rho W_g(P) \big [ \bar P_\rho H_g(P) \bar P_\rho - E_g(P) + \varepsilon \big ]^{-1} \bar P_\rho W_g(P) P_\rho. \label{eq:F(epsilon)_intro}
\end{align}

By the functional calculus, the projection onto the eigenspace associated with the eigenvalue $E_g(P)$ of $H_g(P)$ can be written as $\mathds{1}_{\{E_g(P)\}} ( H_g(P) )$. The limit $F_\rho(0) := \lim_{\varepsilon \to 0} F_\rho(\varepsilon)$ being well-defined (in the norm topology), we will deduce from \eqref{eq:Feshbach_identity} that
\begin{equation}
F_\rho(0) P_\rho \mathds{1}_{\{E_g(P)\}} ( H_g(P) ) P_\rho = 0.
\end{equation}
Choosing $\rho = |g|^{2 - 2\tau}$ with $\tau > 0$ small enough and using the property that
\begin{equation}
\| \bar P_0 \mathds{1}_{\{E_g(P)\}} ( H_g(P) ) \| \le \mathrm{Const} \, g^2,
\end{equation}
(see Appendix \ref{section:estimates}), we will obtain that
\begin{equation}\label{eq:PFP_intro}
P_0 F_\rho(0) P_0 \mathds{1}_{\{E_g(P)\}} ( H_g(P) ) P_0 = O( |g|^{4-2\tau} ).
\end{equation}
Assuming by contradiction that
\begin{equation}\label{eq:intro_contradiction}
\mathrm{dim} \, \mathrm{Ker} \, ( H_g(P) - E_g(P) ) = 4,
\end{equation}
we will then show that \eqref{eq:PFP_intro} implies
\begin{equation}\label{eq:PFP_intro2}
P_0 F_\rho(0) P_0 = O( |g|^{4-2\tau} ).
\end{equation}

Next, for $\rho = |g|^{2-2\tau} \gg g^2$, using standard estimates involving creation and annihilation operators, it can be verified, in a way similar to \cite{BFS1,BFS2}, that the reduced resolvent $[ \bar P_\rho H_g(P) \bar P_\rho - E_g(P) ]^{-1} \bar P_\rho$  decomposes into the Neumann series
\begin{align}
\big [ \bar P_\rho H_g(P) \bar P_\rho - E_g(P) \big ]^{-1} \bar P_\rho = & \big [ H_0(P) - E_g(P) \big ]^{-1} \notag \\
& \sum_{ n \ge 0 } \big ( - \bar P_\rho W_g(P) \bar P_\rho \big [ H_0(P) - E_g(P) \big ]^{-1} \big )^n \bar P_\rho. \label{eq:Neumann_intro}
\end{align}
Introducing \eqref{eq:Neumann_intro} into \eqref{eq:F(epsilon)_intro} and \eqref{eq:PFP_intro2}, we will deduce the following identity:
\begin{align}
& \sum_{ n \ge 0 } P_0 W_g(P) \big [ H_0(P) - E_g(P) \big ]^{-1} \big ( - \bar P_\rho W_g(P) \bar P_\rho \big [ H_0(P) - E_g(P) \big ]^{-1} \big )^n \bar P_0 W_g(P) P_0 \notag \\
& = (E_0(P) - E_g(P)) P_0 + O(|g|^{2+\tau}). \label{eq:diagonal}
\end{align}
Hence, identifying the left-hand-side of \eqref{eq:diagonal} with a $4\times 4$ matrix, \eqref{eq:diagonal} implies in particular that all the terms of order $g^2$ must be located on the diagonal. Extracting the latter from the sum over $n$ will lead to a contradiction.

\subsection{Organization of the paper}

We decompose the proof of Theorem \ref{thm:main} into two main steps. In Section \ref{section:Feshbach}, we introduce and study some properties of the Feshbach-Schur operator $F_\rho(\varepsilon)$ mentioned in the previous subsection. Next, in Section \ref{section:proof}, we assume that the multiplicity of $E_g(P)$ is equal to $4$, and we conclude the proof of Theorem \ref{thm:main} by a contradiction argument. In Appendix \ref{section:estimates}, we collect some technical estimates used in Sections \ref{section:Feshbach} and \ref{section:proof}. 

Throughout the paper, $\mathrm{C}, \mathrm{C}', \mathrm{C}''$ will denote positive constants that may differ from one line to another.

\section{The Feshbach-Schur operator}\label{section:Feshbach}

As mentioned above, it is convenient to work with the Hamiltonian $\tilde H_g(P)$ obtained from $H_g(P)$ by Wick ordering, that is $\tilde H_g(P) = \, ~ : H_g(P) : ~$, with the usual notations. It is not difficult to check that $\tilde H_g(P) = H_g(P) - g^2 \mathrm{C}_\Lambda$, where $\mathrm{C}_\Lambda$ is a positive constant depending on the ultraviolet cutoff parameter $\Lambda$. Hence it suffices to prove Theorem \ref{thm:main} with $\tilde H_g(P)$ replacing $H_g(P)$ and $\tilde E_g(P) := \inf \sigma( \tilde H_g(P) )$ replacing $E_g(P)$. To simplify the notations, we redefine $H_g(P) := \tilde H_g(P)$ and $E_g(P) := \tilde E_g(P)$. Moreover, in what follows, we drop the dependence on $P$ everywhere unless a confusion may arise. In particular, we set
\begin{align}
& H_g = H_{g}(P), \quad H_0 = H_{0}(P), \quad W_g = W_g(P), \notag \\
& E_g = E_{g}(P) , \quad E_0 = E_0(P) = e_0 + \frac{P^2}{2M}.
\end{align}

Let us begin with the proof of the convergence of the Neumann series \eqref{eq:Neumann_intro}.
\begin{lemma}\label{lm:Neumann}
There exist $g_c>0$ and $p_c>0$ such that, for all $0\le |g|\le g_c$, $0 \le |P| \le p_c$, $\varepsilon \ge 0$ and $g^2 \ll \rho \ll 1$, the operator
$
\bar P_\rho H_g \bar P_\rho - E_g + \varepsilon : D(H_0) \cap \mathrm{Ran}( \bar P_\rho ) \to \mathrm{Ran} ( \bar P_\rho )$ is invertible and satisfies
\begin{align}
& \big [ \bar P_\rho H_g \bar P_\rho - E_g + \varepsilon \big ]^{-1} \bar P_\rho \notag \\
& = \big [ H_0 - E_g + \varepsilon \big ]^{-1} \bar P_\rho \sum_{n\ge0} \left ( - W_g \bar P_\rho \big [ H_0 - E_g + \varepsilon \big ]^{-1} \bar P_\rho \right )^n. \label{eq:Neumann}
\end{align}
\end{lemma}
\begin{proof}
Since
\begin{equation}
\bar P_\rho \big ( H_g - E_g + \varepsilon \big ) \bar P_\rho = \big ( H_0 - E_g + \varepsilon ) \bar P_\rho + \bar P_\rho W_g \bar P_\rho,
\end{equation}
it suffices to prove that the Neumann series in the right-hand-side of \eqref{eq:Neumann} is convergent. It follows from Lemmata \ref{lm:H_0_barPrho} and \ref{lm:estimate_Wg} in Appendix \ref{section:estimates} that, for all $n \in \mathbb{N}$, $\varepsilon \ge 0$ and $\rho >0$,
\begin{align}
& \left \| \big [ H_0 - E_g + \varepsilon \big ]^{-1} \bar P_\rho \left ( - W_g \bar P_\rho \big [ H_0 - E_g + \varepsilon \big ]^{-1} \bar P_\rho \right )^n \right \| \notag \\
& \le \mathrm{C} \rho^{-1} \big ( \mathrm{C}' |g| \rho^{-\frac{1}{2}} \big )^n, \label{eq:estimate_Neumann}
\end{align} 
Therefore, for $1 \gg \rho \gg g^2$, \eqref{eq:estimate_Neumann} implies \eqref{eq:Neumann}.
\end{proof}

Using Lemma \ref{lm:Neumann}, we now verify that the Feshbach-Schur operator $F_\rho(\varepsilon)$ written in \eqref{eq:F(epsilon)_intro} is well-defined for any $\varepsilon \ge 0$. 
\begin{lemma}\label{lm:existenceF(epsilon)}
There exist $g_c>0$ and $p_c>0$ such that, for all $0\le |g|\le g_c$, $0 \le |P| \le p_c$, $\varepsilon \ge 0$ and $g^2 \ll \rho \ll 1$, the Feshbach-Schur operator
\begin{align}
F_\rho(\varepsilon) 
&= ( H_0 - E_g + \varepsilon ) P_\rho + P_\rho W_g P_\rho - P_\rho W_g \big [ \bar P_\rho H_g \bar P_\rho - E_g + \varepsilon \big ]^{-1} \bar P_\rho W_g P_\rho. \label{eq:F(epsilon)}
\end{align}
is a well-defined (bounded) operator on $\mathrm{Ran}( P_\rho )$. Moreover, $F_\rho(\varepsilon)$ satisfies
\begin{equation}\label{eq:F(epsilon)_to_F(0)}
F_\rho(0) = \lim_{\varepsilon \to 0^+} F_\rho(\varepsilon),
\end{equation}
in the norm topology, and
\begin{equation}\label{eq:||F(0)||}
\| F_\rho(0) \| \le \mathrm{C} \rho.
\end{equation}
\end{lemma}
\begin{proof}
By Lemma \ref{lm:Neumann} and the fact that $\mathrm{Ran}( P_\rho ) \subset D(H_0) \subset D( W_g )$, $F_\rho(\varepsilon)$ is obviously well-defined on $\mathrm{Ran}( P_\rho )$, for any $\varepsilon \ge 0$. The boundedness of $F_\rho(\varepsilon)$ and Equation \eqref{eq:F(epsilon)_to_F(0)} are straightforward verifications.

In order to prove \eqref{eq:||F(0)||}, we proceed as follows: First, it follows from Lemma \ref{lm:Eg-E0} that
\begin{align}
\big \| ( H_0 - E_g ) P_\rho \big \| & = \big \| ( E_0 - E_g )P_\rho + ( - \frac{1}{M} P \cdot P_\ph + \frac{1}{2M} P_\ph^2 + H_\ph ) P_\rho \| \notag \\
& \le \mathrm{C} g^2 + \mathrm{C}' \rho \le \mathrm{C}'' \rho, \label{eq:||F(0)||_1}
\end{align}
since, by assumption, $\rho \gg g^2$. Next, by Lemma \ref{lm:estimate_Wg}, we have that
\begin{equation}\label{eq:||F(0)||_2}
\| P_\rho W_g P_\rho \| \le \mathrm{C} |g| \rho^{\frac{1}{2}} \le \mathrm{C}' \rho.
\end{equation}
Lemma \ref{lm:Neumann} gives
\begin{align}
& P_\rho W_g \big [ \bar P_\rho H_g \bar P_\rho - E_g \big ]^{-1} \bar P_\rho W_g P_\rho \notag \\
& = P_\rho W_g \big [ H_0 - E_g  \big ]^{-1} \bar P_\rho \sum_{n\ge0} \left ( - W_g \bar P_\rho \big [ H_0 - E_g \big ]^{-1} \bar P_\rho \right )^n W_g P_\rho.
\end{align}
Using again Lemma \ref{lm:estimate_Wg}, we get, for all $n \ge 0$,
\begin{align}
& \big \| P_\rho W_g \big [ H_0 - E_g  \big ]^{-1} \bar P_\rho \left ( - W_g \bar P_\rho \big [ H_0 - E_g \big ]^{-1} \bar P_\rho \right )^n W_g P_\rho \big \| \notag \\
& \le \mathrm{C} g^2 ( \mathrm{C}' |g| \rho^{-\frac{1}{2}} )^n, \label{eq:estimate_Neumann_2}
\end{align}
which implies
\begin{equation}\label{eq:||F(0)||_3}
\big \| P_\rho W_g \big [ \bar P_\rho H_g \bar P_\rho - E_g \big ]^{-1} \bar P_\rho W_g P_\rho \big \| \le \mathrm{C} g^2 \le \mathrm{C}' \rho.
\end{equation}
Equations \eqref{eq:||F(0)||_1}, \eqref{eq:||F(0)||_2} and \eqref{eq:||F(0)||_3} give \eqref{eq:||F(0)||}.
\end{proof}

We now turn to the proof of the Feshbach-Schur identity, Equation \eqref{eq:Feshbach_identity}. We refer to \cite{BFS1,BCFS1,GH} for definitions and properties of the ``(smooth) Feshbach-Schur map'', and its use in the context of non-relativistic QED.  In our case, the operator $H_g - E_g + \varepsilon$ is obviously invertible (for $\varepsilon>0$), so that the following theorem simply follows from usual second order perturbation theory. For the convenience of the reader, we recall the proof.
\begin{lemma}\label{lm:Feshbach}
There exist $g_c>0$ and $p_c>0$ such that, for all $0\le |g|\le g_c$, $0 \le |P| \le p_c$, $\varepsilon > 0$ and $g^2 \ll \rho \ll 1$, the operators $H_g - E_g + \varepsilon : D( H_0 ) \to \mathbb{C}^4 \otimes \mathrm{L}^2( \mathbb{R}^3 ) \otimes \mathcal{H}_{\mathrm{ph}}$ and $F_\rho( \varepsilon ) : \mathrm{Ran}( P_\rho ) \to \mathrm{Ran}( P_\rho )$ are invertible and satisfy
\begin{equation}\label{eq:Feshbach}
P_\rho [ H_g - E_g + \varepsilon ]^{-1} P_\rho = F_\rho (\varepsilon)^{-1}.
\end{equation}
\end{lemma}
\begin{proof}
Since $H_g - E_g \ge 0$, for any $\varepsilon>0$, the operator $H_g - E_g + \varepsilon$  from $D(H_0)$ to $\mathbb{C}^4 \otimes \mathrm{L}^2( \mathbb{R}^3 , \d r ) \otimes \mathcal{H}_{\mathrm{ph}}$ is obviously invertible. Next, since $H_0$ commutes with $P_\rho$, we have that
\begin{equation}\label{eq:Wick2}
\big ( H_0 - E_g + \varepsilon \big ) P_\rho + P_\rho W_g P_\rho = P_\rho \big ( H_g - E_g + \varepsilon \big ) P_\rho.
\end{equation}
Combining \eqref{eq:Wick2} with the facts that $\bar P_\rho W_g P_\rho = \bar P_\rho ( H_g - E_g + \varepsilon ) P_\rho$ and $P_\rho + \bar P_\rho = \mathds{1}$, we get
\begin{align}
& F_\rho ( \varepsilon ) P_\rho \big [ H_g - E_g + \varepsilon \big ]^{-1} P_\rho \notag \\
& = P_\rho \big ( H_g - E_g + \varepsilon \big ) P_\rho \big [ H_g - E_g + \varepsilon \big ]^{-1} P_\rho \notag \\
& \quad - P_\rho W_g [ \bar P_\rho H_g \bar P_\rho - E_g + \varepsilon ]^{-1} \bar P_\rho \big ( H_g - E_g + \varepsilon \big ) P_\rho \big [ H_g - E_g + \varepsilon \big ]^{-1} P_\rho \notag \\
& = P_\rho - P_\rho \big ( H_g - E_g + \varepsilon \big ) \bar P_\rho \big [ H_g - E_g + \varepsilon \big ]^{-1} P_\rho \notag \\
& \quad + P_\rho W_g [ \bar P_\rho H_g \bar P_\rho - E_g + \varepsilon ]^{-1} \bar P_\rho \big ( H_g - E_g + \varepsilon \big ) \bar P_\rho \big [ H_g - E_g + \varepsilon \big ]^{-1} P_\rho \notag \\
& = P_\rho - P_\rho W_g \bar P_\rho \big [ H_g - E_g + \varepsilon \big ]^{-1} P_\rho + P_\rho W_g \bar P_\rho \big [ H_g - E_g + \varepsilon \big ]^{-1} P_\rho \notag \\
& = P_\rho.
\end{align}
The identity $P_\rho [ H_g - E_g + \varepsilon ]^{-1} P_\rho F_\rho ( \varepsilon ) = P_\rho$ follows similarly, which proves \eqref{eq:Feshbach}.
\end{proof}
As a consequence of Lemmata \ref{lm:existenceF(epsilon)} and \ref{lm:Feshbach}, we obtain the following lemma.
\begin{lemma}\label{lm:F(0)=0}
There exist $g_c>0$ and $p_c>0$ such that, for all $0\le |g|\le g_c$, $0 \le |P| \le p_c$, and $g^2 \ll \rho \ll 1$, 
\begin{equation}\label{eq:Feshbach_eps_to_0}
F_\rho(0) P_\rho \mathds{1}_{\{E_g\}}(H_g) P_\rho = 0.
\end{equation}
\end{lemma}
\begin{proof}
We obtain from \eqref{eq:Feshbach} that
\begin{equation}\label{eq:Feshbach_bis}
F_\rho (\varepsilon) P_\rho [ H_g - E_g + \varepsilon ]^{-1} P_\rho = P_\rho.
\end{equation}
for all $\varepsilon > 0$. It follows from the functional calculus that
\begin{equation}
 \mathrm{s}-\lim_{\varepsilon \to 0^+} \varepsilon [ H_g - E_g + \varepsilon ]^{-1} = \mathds{1}_{\{E_g\}}(H_g),
\end{equation}
where $\mathrm{s}-\lim$ stands for strong limit. Hence, using \eqref{eq:F(epsilon)_to_F(0)}, we obtain \eqref{eq:Feshbach_eps_to_0} by multiplying \eqref{eq:Feshbach_bis} by $\varepsilon$ and letting $\varepsilon$ go to 0.
\end{proof}
Introducing \eqref{eq:Neumann} into \eqref{eq:F(epsilon)}, we obtain the following identity which holds for any $\varepsilon \ge 0$ and $g^2 \ll \rho \ll 1$: 
\begin{align}
& F_\rho (\varepsilon) = \big ( H_0 - E_g + \varepsilon \big ) P_\rho + P_\rho W_g P_\rho \notag \\
& - \sum_{ n \ge 0 } P_\rho W_g \big [ H_0 - E_g + \varepsilon \big ]^{-1}  \big ( - \bar P_\rho W_g \bar P_\rho \big [ H_0 - E_g + \varepsilon \big ]^{-1} \big )^n \bar P_\rho W_g P_\rho. \label{eq:F(epsilon)_2}
\end{align}
The next lemma will be used in the proof of Theorem \ref{thm:main}.
\begin{lemma}\label{lm:term_order_g^2}
There exist $g_c>0$ and $p_c>0$ such that, for all $0\le|g| \le g_c$ and $0\le|P|\le p_c$,
\begin{align}
& P_0 F_\rho (0) P_0 = \big ( E_0 - E_g \big ) P_0 \notag \\
&\quad - \sum_{\lambda=1,2} \int_{\mathbb{R}^3} P_0 \tilde w(r,k,\lambda) \big [ H_r + \frac{1}{2M} ( P - k )^2 + |k| - E_g \big ]^{-1} w(r,k,\lambda) P_0 \d k \notag \\
&\quad + O( |g|^{2+\tau} ), \label{eq:F(epsilon)_3}
\end{align}
where $\rho = |g|^{2-2\tau}$, $\tau>0$ is fixed sufficiently small, and
\begin{align}
w(r,k,\lambda) := & - \frac{ g }{ m_\el } \left ( \big ( \frac{ m_\el }{ M } ( P - P_\ph ) + p_r \big ) \cdot h^A ( \frac{ m_\el }{ M } g^{\frac{2}{3}} r , k , \lambda ) \right ) \notag \\
& + \frac{ g }{ m_\n } \left ( \big ( \frac{ m_\n }{ M } ( P - P_\ph ) - p_r \big ) \cdot h^A ( - \frac{ m_\n }{ M } g^{\frac{2}{3}} r , k , \lambda ) \right ) \notag \\
& - \frac{g}{2m_\el} \sigma^{\el} \cdot h^B( \frac{ m_\el }{ M } g^{\frac{2}{3}} r , k , \lambda ) + \frac{g}{2m_\n} \sigma^{\n} \cdot h^B ( - \frac{ m_\n }{ M } g^{\frac{2}{3}} r , k , \lambda ), \label{eq:w(r,k,lambda)}
\end{align}
respectively
\begin{align}
\tilde w(r,k,\lambda) := & - \frac{ g }{ m_\el } \left ( \big ( \frac{ m_\el }{ M } ( P - P_\ph ) + p_r \big ) \cdot \bar{h}^A ( \frac{ m_\el }{ M } g^{\frac{2}{3}} r , k , \lambda ) \right ) \notag \\
& + \frac{ g }{ m_\n } \left ( \big ( \frac{ m_\n }{ M } ( P - P_\ph ) - p_r \big ) \cdot \bar{h}^A ( - \frac{ m_\n }{ M } g^{\frac{2}{3}} r , k , \lambda ) \right ) \notag \\
& - \frac{g}{2m_\el} \sigma^{\el} \cdot \bar{h}^B( \frac{ m_\el }{ M } g^{\frac{2}{3}} r , k , \lambda ) + \frac{g}{2m_\n} \sigma^{\n} \cdot \bar{h}^B ( - \frac{ m_\n }{ M } g^{\frac{2}{3}} r , k , \lambda ). \label{eq:tilde_w(r,k,lambda)}
\end{align}
\end{lemma}
\begin{proof}
Since $P_0 P_\rho = P_\rho P_0 = P_0$ and $H_0 P_0 = E_0 P_0$, we have that
\begin{align}
& P_0 F_\rho(0) P_0 = \big ( E_0 - E_g \big ) P_0 + P_0 W_g P_0 - P_0 W_g \big [ H_0 - E_g \big ]^{-1} \bar P_\rho W_g P_0 \notag \\
& - \sum_{ n \ge 1 } P_0 W_g \big [ H_0 - E_g \big ]^{-1}  \big ( - \bar P_\rho W_g \bar P_\rho \big [ H_0 - E_g \big ]^{-1} \big )^n \bar P_\rho W_g P_0.
\end{align}
By \eqref{eq:Wg_Wick}, $P_0 W_g P_0 = 0$. Moreover, using \eqref{eq:estimate_Neumann_2} for $n \ge 1$, we obtain that 
\begin{align}
P_0 F_\rho (0) P_0 &= \big ( E_0 - E_g \big ) P_0 - P_0 W_g \big [ H_0 - E_g \big ]^{-1} \bar P_\rho W_g P_0 + O( |g|^3 \rho^{-\frac{1}{2}} ). 
\end{align}
We conclude the proof by applying Lemma \ref{lm:order_g2} of Appendix \ref{section:estimates}.
\end{proof}
%
%

\section{Proof of Theorem \ref{thm:main}}\label{section:proof}
From now on we assume that $\mathrm{dim} \, \mathrm{Ker} \, ( H_g - E_g ) = 4$, which will lead to a contradiction at the end of this section.
\begin{lemma}\label{lm:P0PgP0_invertible}
There exist $g_c>0$ and $p_c>0$ such that, for all $0\le|g|\le g_c$ and $0\le|P|\le p_c$, the following holds: If $\mathrm{dim} \, \mathrm{Ker} \, ( H_g - E_g ) = 4$, then $P_0 \mathds{1}_{ \{E_g\} }( H_g ) P_0$ is invertible on $\mathrm{Ran}(P_0)$ and satisfies
\begin{equation}\label{eq:||P0PgP0^-1||}
\big \|[ P_0 \mathds{1}_{ \{E_g\} }( H_g ) P_0 ]^{-1} \big \| \le \frac{1}{1 - \mathrm{C} g^2 }.
\end{equation}
\end{lemma}
\begin{proof}
In order to prove that $P_0 \mathds{1}_{ \{E_g\} }( H_g ) P_0$ is invertible on $\mathrm{Ran}(P_0)$, it suffices to show that 
\begin{equation}\label{P0PgP0_invertible_1}
\big \| P_0 - P_0 \mathds{1}_{ \{E_g\} }( H_g ) P_0 \big \| < 1.
\end{equation}
Observe that $P_0 - P_0 \mathds{1}_{ \{E_g\} }( H_g ) P_0$ is a finite rank and positive operator. We have that
\begin{align}
\big \| P_0 - P_0 \mathds{1}_{ \{E_g\} }( H_g ) P_0 \big \| &\le \mathrm{tr}( P_0 - P_0 \mathds{1}_{ \{E_g\} }( H_g ) P_0 ) \notag \\
& = \mathrm{tr} ( P_0 ) - \mathrm{tr} ( P_0 \mathds{1}_{ \{E_g\} }( H_g ) ) \notag \\
& = \mathrm{tr} ( P_0 ) - \mathrm{tr}( \mathds{1}_{ \{E_g\} }( H_g ) ) + \mathrm{tr}( \bar P_0 \mathds{1}_{ \{E_g\} }( H_g ) ) \notag \\
& = 4 - 4 + \mathrm{tr}( \bar P_0 \mathds{1}_{ \{E_g\} }( H_g ) ) = \mathrm{tr}( \bar P_0 \mathds{1}_{ \{E_g\} }( H_g ) ). \label{P0PgP0_invertible_2}
\end{align}
The projection $\bar P_0$ can be decomposed as
\begin{equation}\label{P0PgP0_invertible_3}
\bar P_0 = \mathds{1} \otimes \bar P_{\phi_0} \otimes P_\Omega + \mathds{1} \otimes \mathds{1} \otimes \bar P_\Omega.
\end{equation}
It follows from Lemma \ref{lm:Pphi0_Pg} that 
\begin{equation}\label{P0PgP0_invertible_4}
\mathrm{tr}( ( \mathds{1} \otimes \bar P_{\phi_0} \otimes P_\Omega ) \mathds{1}_{ \{E_g\} }( H_g ) ) \le \mathrm{C} g^2,
\end{equation}
and from Lemma \ref{lm:barP0_Pg} that
\begin{equation}\label{P0PgP0_invertible_5}
\mathrm{tr}( ( \mathds{1} \otimes \mathds{1} \otimes \bar P_\Omega ) P_g ) \le \mathrm{tr}( \mathcal{N}_\ph P_g ) \le \mathrm{C} g^2.
\end{equation}
Therefore, $\| P_0 - P_0 \mathds{1}_{ \{E_g\} }( H_g ) P_0 \| \le \mathrm{C} g^2$. The invertibility of $P_0 \mathds{1}_{ \{E_g\} }( H_g ) P_0$ and Equation \eqref{eq:||P0PgP0^-1||} directly follow from the latter estimate.
\end{proof}
As a consequence of Lemma \ref{lm:P0PgP0_invertible}, we obtain the following lemma.
\begin{lemma}\label{lm:term_order_g^2_1}
Let $\Gamma$ denote the operator on $\mathrm{Ran}(P_0)$ defined by
\begin{equation}
\Gamma := \sum_{\lambda=1,2} \int_{\mathbb{R}^3} P_0 \tilde w(r,k,\lambda) \big [ H_r + \frac{1}{2M} ( P - k )^2 + |k| - E_g \big ]^{-1} w(r,k,\lambda) P_0 \d k, \label{eq:def_M}
\end{equation}
with $w(r,k,\lambda)$ and $\tilde w(r,k,\lambda)$ as in \eqref{eq:w(r,k,lambda)}--\eqref{eq:tilde_w(r,k,lambda)}. There exist $g_c>0$ and $p_c>0$ such that, for all $0\le |g| \le g_c$ and $0\le|P|\le p_c$, the following holds: If $\mathrm{dim} \, \mathrm{Ker} \, ( H_g - E_g ) = 4$, then 
\begin{align}
\Gamma = \big ( E_0 - E_g \big ) P_0 + O( |g|^{2+\tau} ), \label{eq:term_order_g^2_1}
\end{align}
where $\tau>0$ is fixed sufficiently small.
\end{lemma}
\begin{proof}
Fix $\rho = |g|^{2-2\tau}$ for some sufficiently small $\tau>0$. Multiplying both sides of Equation \eqref{eq:Feshbach_eps_to_0} by $P_0$, we get
\begin{equation}\label{eq:term_order_g^2_2}
P_0 F_\rho(0) P_\rho \mathds{1}_{ \{E_g \} }( H_g ) P_0 = 0.
\end{equation}
Introducing the decomposition $\mathds{1} = P_0 + \bar P_0$ into \eqref{eq:term_order_g^2_2} and using Lemma \ref{lm:P0PgP0_invertible}, this yields
\begin{equation}\label{eq:term_order_g^2_3}
P_0 F_\rho(0) P_0 = - P_0 F_\rho(0) P_\rho \bar P_0 \mathds{1}_{ \{E_g \} }( H_g ) P_0 [ P_0 \mathds{1}_{ \{E_g \} }( H_g ) P_0 ]^{-1}. 
\end{equation}
By Equations \eqref{P0PgP0_invertible_3}, \eqref{P0PgP0_invertible_4} and \eqref{P0PgP0_invertible_5}, we learn that
\begin{equation}
\big \| \bar P_0 \mathds{1}_{ \{E_g \} }( H_g ) \big \| \le \mathrm{tr}( \bar P_0 \mathds{1}_{ \{E_g \} }( H_g ) ) \le \mathrm{C} g^2,
\end{equation}
which, combined with \eqref{eq:||F(0)||} and \eqref{eq:||P0PgP0^-1||}, implies that
\begin{equation}
\big \|P_0 F_\rho(0) P_\rho \bar P_0 \mathds{1}_{ \{E_g \} }( H_g ) P_0 [ P_0 \mathds{1}_{ \{E_g \} }( H_g ) P_0 ]^{-1} \big \| \le \mathrm{C} g^2 \rho = \mathrm{C} |g|^{4 - 2\tau}.
\end{equation}
We conclude the proof thanks to Lemma \ref{lm:term_order_g^2}.
\end{proof}
Let us consider the canonical orthonormal basis of $\mathbb{C}^4$ in which the Pauli matrices $\sigma^\el_j$, $\sigma^\n_j$, $j \in \{1,2,3\}$, are given by \eqref{eq:sigma_el}--\eqref{eq:sigma_n}. Obviously, $\Gamma$ identifies with a $4\times 4$ matrix in this basis. In the next theorem, we determine a non-diagonal coefficient of $\Gamma$ of the form $-\mathrm{C}_0 g^2 + o(g^2)$ with $\mathrm{C}_0 > 0$.
\begin{theorem}\label{thm:M32}
Let $\Gamma$ be given as in \eqref{eq:def_M}. There exist $g_c>0$ and $p_c>0$ such that, for all $0\le |g| \le g_c$ and $0\le|P|\le p_c$, the coefficient of $\Gamma$ located on the third line and second column, $\Gamma_{32}$, satisfies
\begin{equation}\label{eq:M32}
\Gamma_{32} = - \mathrm{C}_0 g^2 + O( |g|^{ \frac{8}{3} } ),
\end{equation}
where $\mathrm{C}_0$ is a strictly positive constant independent of $g$.
\end{theorem}
\begin{proof}
We view $w(r,k,\lambda)$ as a linear combination (some coefficients being given by operators) of the functions $h^A_j( \cdots )$ and $h^B_j( \cdots )$, $j \in \{ 1 ,2 ,3 \}$. We introduce the corresponding expression into \eqref{eq:def_M} and consider each term separately. 

Since the coefficients located on the third line and second column of the Pauli matrices expressed in \eqref{eq:sigma_el}--\eqref{eq:sigma_n} vanish, the terms containing at least one factor $h^A_j( \cdots )$ do not contribute to $\Gamma_{32}$. The same holds for the terms containing at least one factor $h^B_3( \cdots )$, since the third Pauli matrices, $\sigma^\el_3$ and $\sigma^\n_3$, are diagonal.

Therefore, $\Gamma_{32}$ is equal to the coefficient located on the third line and second column of the matrix $\Gamma'$ given by
\begin{align}
\Gamma' = \sum_{\lambda=1,2} & \int_{\mathbb{R}^3} P_0 \sum_{j=1,2} \Big ( - \frac{g}{2m_\el} \sigma^{\el}_j \bar{h}^B_j( \frac{ m_\el }{ M } g^{\frac{2}{3}} r , k , \lambda ) + \frac{g}{2m_\n} \sigma^{\n}_j \bar{h}^B_j ( - \frac{ m_\n }{ M } g^{\frac{2}{3}} r , k , \lambda ) \Big ) \notag \\
& \Big [ H_r + \frac{1}{2M} ( P - k )^2 + |k| - E_g \Big ]^{-1} \phantom{\sum_j} \notag \\
& \sum_{j'=1,2} \Big ( - \frac{g}{2m_\el} \sigma^{\el}_{j'} h^B_{j'}( \frac{ m_\el }{ M } g^{\frac{2}{3}} r , k , \lambda ) + \frac{g}{2m_\n} \sigma^{\n}_{j'} h^B_{j'} ( - \frac{ m_\n }{ M } g^{\frac{2}{3}} r , k , \lambda ) \Big ) P_0 \d k. \label{eq:def_M'}
\end{align}
It follows from the definition \eqref{eq:h^B} of $h^B_j$ that
\begin{equation}\label{eq:h(r)-h(0)}
\big | h^B_j( r, k ,\lambda ) - h^B_j( 0 , k, \lambda ) \big | \le \mathrm{C} |k|^{\frac{3}{2}} \chi_\Lambda(k) |r|,
\end{equation}
for any $j \in \{1,2,3\}$, $\lambda \in \{1,2\}$, $r \in \mathbb{R}^3$ and $k \in \mathbb{R}^3$. Moreover, the expression \eqref{eq:phi0} of $\phi_0$ implies that
\begin{equation}\label{eq:rphi0}
\big \| |r| \phi_0(r) \big \| \le \mathrm{C}.
\end{equation}
Hence, using in addition that, for $|P|$ sufficiently small,
\begin{equation}
\Big \| \Big [ H_r + \frac{ ( P - k )^2 }{2M} + |k| - E_g \Big ]^{-1} \Big \| \le \frac{ \mathrm{C} }{ |k| }, \label{eq:resolvante_Hr}
\end{equation}
we obtain from \eqref{eq:def_M'} and \eqref{eq:h(r)-h(0)}--\eqref{eq:resolvante_Hr} that
\begin{align}
\Gamma' = \sum_{\lambda=1,2} & \int_{\mathbb{R}^3} P_0 \sum_{j=1,2} \Big ( - \frac{g}{2m_\el} \sigma^{\el}_j \bar{h}^B_j( 0 , k , \lambda ) + \frac{g}{2m_\n} \sigma^{\n}_j \bar{h}^B_j ( 0 , k , \lambda ) \Big ) \notag \\
& \Big [ e_0 + \frac{1}{2M} ( P - k )^2 + |k| - E_g \Big ]^{-1} \phantom{\sum_j} \notag \\
& \sum_{j'=1,2} \Big ( - \frac{g}{2m_\el} \sigma^{\el}_{j'} h^B_{j'}( 0 , k , \lambda ) + \frac{g}{2m_\n} \sigma^{\n}_{j'} h^B_{j'} ( 0 , k , \lambda ) \Big ) P_0 \d k + O( |g|^{ \frac{8}{3}} ). \label{eq:def_M'_2}
\end{align}

Notice now that, for $j,j' \in \{1,2\}$, the coefficient on the third line and second column of the products $\sigma^\el_j \sigma^\el_{j'}$ and $\sigma^\n_j \sigma^\n_{j'}$ vanishes. We thus obtain from \eqref{eq:def_M'_2} that
\begin{equation}\label{eq:M32_1}
\Gamma_{32} = \Gamma'_{32} = \gamma_1 + \gamma_2 + O( |g|^{ \frac{8}{3}} ),
\end{equation}
where
\begin{align}
\gamma_1 :=& - \frac{ g^2 }{ 4 m_\el m_\n } \sum_{\lambda=1,2} \int_{\mathbb{R}^3} ( \phi_0 , \big [ \bar{h}^B_1( 0 , k , \lambda ) + \i \bar{h}^B_2 ( 0 , k , \lambda ) \big ] \notag \\
&  \Big [ e_0 + \frac{1}{2M} ( P - k )^2 + |k| - E_g \Big ]^{-1} \big [ h^B_1 ( 0 , k , \lambda ) - \i h^B_2 ( 0 , k , \lambda ) \big ] \phi_0 ) \d k,
\end{align}
and
\begin{align}
\gamma_2 := & - \frac{ g^2 }{ 4 m_\el m_\n } \sum_{\lambda=1,2} \int_{\mathbb{R}^3} ( \phi_0  , \big [ \bar{ h}^B_1 ( 0 , k , \lambda ) - \i \bar{ h}^B_2 ( 0 , k , \lambda ) \big ] \notag \\
& \Big [ e_0 + \frac{1}{2M} ( P - k )^2 + |k| - E_g \Big ]^{-1} \big [ h^B_1( 0 , k , \lambda ) + \i h^B_2 ( 0 , k , \lambda ) \Big ] \phi_0 ) \d k.
\end{align}
We remark that the cross terms involving $h^B_1( 0 , k , \lambda )$ and $h^B_2 ( 0 , k , \lambda )$ vanish. 
Thus, we obtain
\begin{align}
\Gamma_{32} =&  - \frac{ g^2 }{ 2 m_\el m_\n } \sum_{j=1,2} \sum_{\lambda=1,2} \int_{\mathbb{R}^3} \bar{ h }^B_j  ( 0 , k , \lambda ) \notag \\
& \phantom{ - \frac{ g^2 }{ 4 m_\el m_\n } } \Big [ e_0 + \frac{ ( P - k )^2 }{2M} + |k| - E_g \Big ]^{-1} h^B_j ( 0 , k , \lambda ) \d k + O( |g|^{ \frac{8}{3} } ). \label{eq:M32_3}
\end{align}
The integral in the right-hand-side of \eqref{eq:M32_3} still depends on $g$ through the ground state energy $E_g$. Nevertheless, one can readily check that
\begin{align}
& \Big | \Big [ e_0 + \frac{ ( P - k )^2 }{2M} + |k| - E_g \Big ]^{-1} - \Big [ e_0 + \frac{ ( P - k )^2 }{2M} + |k| - E_0 \Big ]^{-1} \Big | \notag \\
& \le | E_0 - E_g | \frac{ \mathrm{C} }{ |k|^2 } \le \frac{ \mathrm{C}'g^2 }{ |k|^2 }, 
\end{align}
where, in the last inequality, we used Lemma \ref{lm:Eg-E0}. Therefore, since, for any $j \in \{ 1,2 \}$ and $\lambda \in \{1,2\}$, the functions $h^B_j ( 0 , k , \lambda )$ satisfy $| h^B_j ( 0 , k , \lambda ) | \le \mathrm{C} |k|^{1/2} \chi_\Lambda(k)$, we get
\begin{align}
\Gamma_{32} =&  - \frac{ g^2 }{ 2 m_\el m_\n } \sum_{j=1,2} \sum_{\lambda=1,2} \int_{\mathbb{R}^3} \bar{ h }^B_j ( 0 , k , \lambda ) \notag \\
& \phantom{ - \frac{ g^2 }{ 4 m_\el m_\n } } \Big [ e_0 + \frac{ ( P - k )^2 }{2M} + |k| - E_0 \Big ]^{-1} h^B_j ( 0 , k , \lambda ) \d k + O( |g|^{ \frac{8}{3} } ). \label{eq:M32_4}
\end{align}
Now, the integrals in the right-hand-side of \eqref{eq:M32_4} can be explicitly computed, which leads to
\begin{align}
\Gamma_{32} =&  - \frac{ g^2 }{ 8 \pi^2 m_\el m_\n } \int_{\mathbb{R}^3} \frac{ |k| \chi_\Lambda( k )^2 }{ k^2/2M - k \cdot P / M + |k| } \big ( \frac{ k_3^2 }{ |k|^2 } + 1 \big ) \d k + O( |g|^{ \frac{8}{3} } ). \label{eq:M32_5}
\end{align}
The integrand in \eqref{eq:M32_5} is strictly positive (for $P$ sufficiently small), and hence the integral does not vanish. This concludes the proof of the theorem.
\end{proof}
We are now able to prove Theorem \ref{thm:main}: \\

\noindent \emph{Proof of Theorem \ref{thm:main}}. By \cite{AGG2}, we know that $\mathrm{dim} \, \mathrm{Ker} ( H_g - E_g ) \le 4$. Assume by contradiction that $\mathrm{dim} \, \mathrm{Ker} ( H_g - E_g ) = 4$. By Lemma \ref{lm:term_order_g^2_1}, the matrix $\Gamma$ defined in \eqref{eq:def_M} satisfies \eqref{eq:term_order_g^2_1}. In particular, in any basis of $\mathbb{C}^4$, the non-vanishing terms of order $g^2$ of $\Gamma$ are necessarily located on the diagonal. However, according to Theorem \ref{thm:M32}, in the canonical orthonormal basis of $\mathbb{C}^4$ in which the Pauli matrices are given by \eqref{eq:sigma_el}--\eqref{eq:sigma_n}, the non-diagonal coefficient $\Gamma_{32}$ contains a non-vanishing term of order $g^2$. Hence we get a contradiction and the theorem is proven. \qed

\appendix

\section{}\label{section:estimates}

In this appendix, we collect some estimates which were used in Sections \ref{section:Feshbach} and \ref{section:proof}. Some of them are standard (see for instance \cite{BFS1,BFS2}). We begin with two lemmata concerning the non-interacting Hamiltonian $H_0$ defined in \eqref{eq:def_H0}.
\begin{lemma}\label{lm:Hph<H_0}
There exists $p_c>0$ such that for all $0\le|P|\le p_c$,
\begin{equation}
H_\ph \le 2( H_0 - E_0 ).
\end{equation}
\end{lemma}
\begin{proof}
For $j \in \{ 1,2,3 \}$, one can easily verify that $| (P_\ph)_j | \le H_\ph$. Hence, since $E_0 = e_0 + P^2/2M$, we have that
\begin{align}
H_0 &= H_r + \frac{P^2}{2M} - \frac{1}{M} P \cdot P_\ph + \frac{1}{2M} P_\ph^2 + H_\ph \ge E_0 + \frac{1}{2} H_\ph,
\end{align}
for $P$ sufficiently small, which proves the lemma.
\end{proof}
\begin{lemma}\label{lm:H_0_barPrho}
There exists $p_c>0$ such that, for all $0 \le |P| \le p_c$ and $\rho \ge 0$,
\begin{equation}\label{eq:H_0_barPrho}
\bar P_\rho H_0 \bar P_\rho \ge \big ( \frac{P^2}{2M} + \min( e_0 + \frac{\rho}{2} , e_1 ) \big ) \bar P_\rho.
\end{equation}
\end{lemma}
\begin{proof}
Since $P_\rho = \mathds{1} \otimes P_{\phi_0} \otimes \mathds{1}_{H_\ph \le \rho}$ in the tensor product $\mathbb{C}^4 \otimes \mathrm{L}^2( \mathbb{R}^3 ) \otimes \mathcal{H}_\ph$, we can write
\begin{equation}
\bar P_\rho = \mathds{1} - P_\rho = \mathds{1} \otimes \bar P_{\phi_0} \otimes \mathds{1}_{H_\ph \le \rho} + \mathds{1} \otimes \mathds{1} \otimes \mathds{1}_{H_\ph \ge \rho},
\end{equation}
where $\bar P_{\phi_0} = \mathds{1}-P_{\phi_0}$. Since $H_r \bar P_{ \phi_0 } \ge e_1 \bar P_{ \phi_0 }$, we get that 
\begin{align}
& H_0 ( \mathds{1} \otimes \bar P_{\phi_0} \otimes \mathds{1}_{H_\ph \le \rho} ) \ge \big ( e_1 + \frac{P^2}{2M} \big ) ( \mathds{1} \otimes \bar P_{\phi_0} \otimes \mathds{1}_{H_\ph \le \rho} ),
\end{align}
for $P$ small enough. Moreover, by Lemma \ref{lm:Hph<H_0},
\begin{align}
& H_0 ( \mathds{1} \otimes \mathds{1} \otimes \mathds{1}_{H_\ph \ge \rho} ) \ge \big ( e_0 + \frac{P^2}{2M} + \frac{\rho}{2} \big ) ( \mathds{1} \otimes \mathds{1} \otimes \mathds{1}_{H_\ph \ge \rho} ).
\end{align}
Hence \eqref{eq:H_0_barPrho} is proven.
\end{proof}
The proof of the next two lemmata being standard, we omit them.
\begin{lemma}\label{lm:standard_N}
For any $f \in \mathrm{L}^2( \mathbb{R}^3 \times \{ 1,2 \} )$, the operators $a(f) [ \mathcal{N}_\ph \bar P_\Omega ]^{-1/2}$ and $[ \mathcal{N}_\ph \bar P_\Omega ]^{-1/2} a(f) $ extend to bounded operators on $\mathcal{H}_\ph$ satisfying
\begin{align}
& \big \| a(f) [	\mathcal{N}_\ph \bar P_\Omega ]^{- \frac{1}{2} } \big \| \le \| f \|, \label{eq:a(f)leN_1} \\
& \big \| [ \mathcal{N}_\ph \bar P_\Omega ]^{- \frac{1}{2} } a(f) \big \| \le \sqrt{2} \| f \|.  \label{eq:a(f)leN_2}
\end{align}
\end{lemma}
\begin{lemma}\label{lm:standard}
Let $f \in \mathrm{L}^2( \mathbb{R}^3 \times \{ 1,2 \} )$ be such that $(k,\lambda) \mapsto |k|^{-1/2} f(k,\lambda) \in \mathrm{L}^2( \mathbb{R}^3 \times \{ 1,2 \} )$. Then, for any $\rho > 0$, the operators $a(f) [ H_\ph + \rho ]^{-1/2}$ and $[ H_\ph + \rho ]^{-1/2} a(f) $ extend to bounded operators on $\mathcal{H}_\ph$ satisfying
\begin{align}
& \big \| a(f) [	H_\ph + \rho]^{- \frac{1}{2} } \big \| \le \| |k|^{-\frac{1}{2}} f \|, \label{eq:a(f)leHf_1} \\
& \big \| [	H_\ph + \rho]^{- \frac{1}{2} } a(f) \big \| \le \| |k|^{-\frac{1}{2}} f \| + \rho^{-\frac{1}{2}} \| f \|.  \label{eq:a(f)leHf_2}
\end{align}
\end{lemma}
The following lemma is taken from \cite{AGG2}. Its proof is based on a ``pull-through'' formula (see \cite{AGG2}).
\begin{lemma}\label{lm:barP0_Pg}
There exist $g_c>0$ and $p_c>0$ such that, for all $0\le|g|\le g_c$ and $0\le|P|\le p_c$, the following holds:
\begin{equation}
\forall \Phi_g \in \mathrm{Ker}( H_g - E_g ) , \| \Phi_g \| = 1, \text{ we have } ( \Phi_g , \mathcal{N}_\ph \Phi_g ) \le \mathrm{C} g^2, 
\end{equation}
where $\mathrm{C}$ is a positive constant independent of $g$.
\end{lemma}
In the next lemma, we estimate the difference between the ground state energies $E_g = \inf \sigma( H_g )$ and $E_0 = \inf \sigma ( H_0 )$.
\begin{lemma}\label{lm:Eg-E0}
There exist $g_c>0$ and $p_c>0$ such that, for all $0\le|g|\le g_c$ and $0\le|P|\le p_c$,
\begin{equation}\label{eq:Eg-E0}
E_g \le E_0 \le E_g + \mathrm{C} g^2,
\end{equation}
where $\mathrm{C}$ is a positive constant independent of $g$.
\end{lemma}
\begin{proof}
Note that, since the perturbation $W_g$ is Wick-ordered, we have that
\begin{equation}\label{eq:Wg_Wick}
( \mathds{1} \otimes \mathds{1} \otimes P_\Omega ) W_g ( \mathds{1} \otimes \mathds{1} \otimes P_\Omega ) = 0,
\end{equation}
where, recall, $P_\Omega$ denotes the orthogonal projection onto the vector space spanned by the Fock vacuum $\Omega$. Hence, by the Rayleigh-Ritz principle,
\begin{equation}\label{eq:Eg<E0}
E_g \le \big ( ( y \otimes \phi_0 \otimes \Omega ) , H_g ( y \otimes \phi_0 \otimes \Omega ) \big ) = \big ( ( y \otimes \phi_0 \otimes \Omega ) , H_0 ( y \otimes \phi_0 \otimes \Omega ) \big ) = E_0,
\end{equation}
where, as above, $y$ denotes an arbitrary normalized element in $\mathbb{C}^4$.

In order to prove the second inequality in \eqref{eq:Eg-E0}, we use Lemmata \ref{lm:standard_N} and \ref{lm:barP0_Pg}. More precisely, let $\Phi_g \in \mathrm{Ker}( H_g - E_g )$, $\| \Phi_g \| =1$ ($\Phi_g$ exists by \cite{AGG2}). We have
\begin{align}
E_0 - E_g \le ( \Phi_g , (H_0 - H_g) \Phi_g ) =- ( \Phi_g , W_g \Phi_g ).
\end{align}
Recall that $W_g$ is given by the Wick-ordered expression obtained from \eqref{eq:Wg}. We express the latter in terms of operators of creation and annihilation, and estimate each term separately. Consider for instance the term
\begin{equation}\label{eq:Wg_term1}
\frac{ g }{ m_\el } \left ( \big ( \frac{ m_\el }{ M } ( P - P_\ph ) + p_r \big ) \cdot a ( h^A( \frac{ m_\el }{ M } g^{\frac{2}{3}} r ) ) \right ).
\end{equation}
It is not difficult to check that 
\begin{equation}\label{eq:Wg_est1}
(P-P_\ph)^2 \le a H_0 + b \quad\text{and}\quad p_r^2 \le a H_0 + b,
\end{equation}
for some positive constants $a$ and $b$ depending on $\mu$ and $M$. One easily deduces from \eqref{eq:Wg_est1} that
\begin{equation}\label{eq:Eg-E_0_1}
\big \| \big ( \frac{ m_\el }{ M } ( P - P_\ph ) + p_r \big ) \Phi_g \big \| \le \mathrm{C}.
\end{equation}
Moreover, by Lemmata \ref{lm:standard_N} and \ref{lm:barP0_Pg}, we have that
\begin{equation}\label{eq:Eg-E_0_2}
\big \| a ( h^A( \frac{ m_\el }{ M } g^{\frac{2}{3}} r ) ) \Phi_g \big \| \le \mathrm{C} \big \| \mathcal{N}_\ph^{\frac{1}{2}} \Phi_g \big \| \le \mathrm{C}' |g|.
\end{equation}
Equations \eqref{eq:Eg-E_0_1} and \eqref{eq:Eg-E_0_2} imply that
\begin{equation}
\big | ( \Phi_g , \eqref{eq:Wg_term1} \Phi_g ) \big | \le \mathrm{C} g^2,
\end{equation}
and since the other terms in $W_g$ are estimated similarly, this concludes the proof.
\end{proof}
Lemma \ref{lm:barP0_Pg} gives an estimation of the overlap of the ground state $\Phi_g$ of $H_g$ with the Fock vacuum. We also need to estimate the overlap of $\Phi_g$ with the ground state $\phi_0$ of the electronic Hamiltonian $H_r$ in the sense stated in  the following lemma.
\begin{lemma}\label{lm:Pphi0_Pg}
There exist $g_c>0$ and $p_c>0$ such that, for all $0\le|g|\le g_c$ and $0\le|P|\le p_c$, the following holds:
\begin{equation}
\forall \Phi_g \in \mathrm{Ker}( H_g - E_g ) , \| \Phi_g \| = 1, \text{ we have } | ( \Phi_g , ( \mathds{1} \otimes \bar P_{\phi_0} \otimes P_\Omega )  \Phi_g ) | \le \mathrm{C} g^2, 
\end{equation}
where $\mathrm{C}$ is a positive constant independent of $g$.
\end{lemma}
\begin{proof}
Let $\Phi_g$ be a normalized ground state of $H_g$, that is $(H_g - E_g) \Phi_g = 0$, $\| \Phi_g \| = 1$. Since $E_0 - E_g = e_0 + P^2/2M - E_g \ge 0$ by Lemma \ref{lm:Eg-E0}, we have that
\begin{align}
0 & = \big ( \Phi_g , ( \mathds{1} \otimes \bar P_{\phi_0} \otimes P_\Omega ) ( H_g - E_g ) \Phi_g \big ) \phantom{ \frac{P^2}{M} } \notag \\
& = \big ( \Phi_g , ( \mathds{1} \otimes \bar P_{\phi_0} \otimes P_\Omega ) \big ( H_r + \frac{P^2}{2M} - E_g + W_g \big ) \Phi_g ) \notag \\
& \ge \big ( \Phi_g , ( \mathds{1} \otimes \bar P_{\phi_0} \otimes P_\Omega ) ( e_1 - e_0 + W_g ) \Phi_g \big ), \phantom{ \frac{P^2}{M} }
\end{align}
and hence
\begin{align}
( \Phi_g , ( \mathds{1} \otimes \bar P_{\phi_0} \otimes P_\Omega ) \Phi_g ) \le - \frac{1}{e_1-e_0} ( \Phi_g , ( \mathds{1} \otimes \bar P_{\phi_0} \otimes P_\Omega ) W_g \Phi_g ).
\end{align}
We conclude the proof thanks to Lemmata \ref{lm:standard_N} and \ref{lm:barP0_Pg}, by arguing in the same way as in the proof of Lemma \ref{lm:Eg-E0}.
\end{proof}
We now give estimates relating the perturbation $W_g$ to $H_0$.
\begin{lemma}\label{lm:estimate_Wg}
There exist $g_c>0$ and $p_c>0$ such that, for all $0 \le |g| \le g_c$, $0 \le |P| \le p_c$, $0 < \rho \ll 1$ and $\varepsilon \ge 0$, the following estimates hold:
\begin{align}
& \big \| [ H_0 - E_g + \varepsilon ]^{-\frac{1}{2}} \bar P_\rho W_g \bar P_\rho [ H_0 - E_g + \varepsilon ]^{-\frac{1}{2}} \big \| \le \mathrm{C} |g| \rho^{-\frac{1}{2}}, \label{eq:estimate_Wg_1}\\
& \big \| P_\rho W_g \bar P_\rho [ H_0 - E_g + \varepsilon ]^{-\frac{1}{2}} \big \| \le \mathrm{C} |g|, \label{eq:estimate_Wg_2} \\
& \big \| [ H_0 - E_g + \varepsilon ]^{-\frac{1}{2}} \bar P_\rho W_g P_\rho \big \| \le \mathrm{C} |g|, \label{eq:estimate_Wg_3} \\
& \big \| P_\rho W_g P_\rho \big \| \le \mathrm{C} |g| \rho^{\frac{1}{2}}. \label{eq:estimate_Wg_4}
\end{align}
\end{lemma}
\begin{proof}
Let us begin with proving \eqref{eq:estimate_Wg_1}. As in the proof of Lemma \ref{lm:Eg-E0}, we express $W_g$ in terms of creation and annihilation operators from the Wick-ordered expression obtained from \eqref{eq:Wg}, and we estimate each term separately. Let us consider again the term \eqref{eq:Wg_term1} as an example. Using \eqref{eq:Wg_est1}, Lemma \ref{lm:H_0_barPrho}, and the fact that $E_0 \ge E_g$, we obtain
\begin{equation}\label{eq:Wg_est2}
\Big \| [ H_0 - E_g + \varepsilon ]^{-\frac{1}{2}} \bar P_\rho \big ( \frac{ m_\el }{ M } ( P - P_\ph ) + p_r \big )_j \Big \| \le \mathrm{C} \rho^{-\frac{1}{2}}.
\end{equation}
for $j\in\{1,2,3\}$. Next, for $j \in \{1,2,3\}$, Lemma \ref{lm:standard} gives
\begin{equation}\label{eq:Wg_est3}
\big \| a ( h_j^A( \frac{ m_\el }{ M } g^{\frac{2}{3}} r ) ) [ H_\ph + \rho ]^{-1/2} \big \| \le \mathrm{C}, 
\end{equation}
and it follows from Lemmata \ref{lm:Hph<H_0} and \ref{lm:H_0_barPrho} that
\begin{equation}\label{eq:Wg_est4}
\big \| [ H_\ph + \rho ]^{\frac{1}{2}} \bar P_\rho [ H_0 - E_g + \varepsilon ]^{-\frac{1}{2}} \big \| \le \mathrm{C}.
\end{equation}
Using \eqref{eq:Wg_est2}, \eqref{eq:Wg_est3} and \eqref{eq:Wg_est4}, we obtain
\begin{equation}\label{eq:Wg_est5}
\big \| [ H_0 - E_g + \varepsilon ]^{-\frac{1}{2}} \bar P_\rho \eqref{eq:Wg_term1} \bar P_\rho [ H_0 - E_g + \varepsilon ]^{-\frac{1}{2}} \big \| \le \mathrm{C} |g| \rho^{-\frac{1}{2}}.
\end{equation}
The other terms in $W_g$ are estimated similarly, using in particular Estimate \eqref{eq:a(f)leHf_2} (in addition to \eqref{eq:a(f)leHf_1}) for the terms quadratic in the annihilation and creation operators. Hence \eqref{eq:estimate_Wg_1} is proven. In order to prove \eqref{eq:estimate_Wg_2}, \eqref{eq:estimate_Wg_3} and \eqref{eq:estimate_Wg_4}, we proceed similarly, using the further following estimates:
\begin{align}
& \Big \|  \big ( \frac{ m_\el }{ M } ( P - P_\ph ) + p_r \big )_j P_\rho \Big \| \le \mathrm{C}, \label{eq:Wg_est6} \\
& \big \|[ H_\ph + \rho ]^{\frac{1}{2}} P_\rho \big \| \le \mathrm{C} \rho^{\frac{1}{2}}. \label{eq:Wg_est7}
\end{align}
Estimate \eqref{eq:Wg_est6} follows from \eqref{eq:Wg_est1}, and \eqref{eq:Wg_est7} is an obvious consequence of the Spectral Theorem.
\end{proof}
\begin{lemma}\label{lm:order_g2}
There exist $g_c>0$ and $p_c>0$ such that, for all $0 \le |g| \le g_c$, $0 \le |P| \le p_c$, $0 < \rho \ll 1$, and $\varepsilon \ge 0$, we have
\begin{align}
& P_0 W_g \big [ H_0 - E_g \big ]^{-1} \bar P_\rho W_g P_0 \notag \\
& = \sum_{\lambda=1,2} \int_{\mathbb{R}^3} P_0 \tilde w(r,k,\lambda) \big [ H_r + \frac{1}{2M} ( P - k )^2 + |k| - E_g \big ]^{-1} w(r,k,\lambda) P_0 \d k \notag \\
&\quad + O( |g|^3 ) + O( g^2 \rho ), \label{eq:order_g2}
\end{align}
where $w(r,k,\lambda)$ and $\tilde w(r,k,\lambda)$ are defined in \eqref{eq:w(r,k,lambda)}--\eqref{eq:tilde_w(r,k,lambda)}.
\end{lemma}
\begin{proof}
The perturbation $W_g$ appears twice in $P_0 W_g [ H_0 - E_g ]^{-1} \bar P_\rho W_g P_0$. We introduce the expression \eqref{eq:Wg} of $W_g$ into the latter operator, and consider each term separately.

First, the terms containing a creation operator in the ``first'' $W_g$ vanish since $P_0$ projects onto the Fock vaccum. The same holds for the terms containing an annihilation operator in the ``second'' $W_g$.

Next, the terms involving the parts of $W_g$ quadratic in the creation and annihilation operators are (at least) of order $O(|g|^3)$, as follows again from Lemmata \ref{lm:standard} and \ref{lm:estimate_Wg}.

Therefore, one can compute
\begin{align}
& P_0 W_g \big [ H_0 - E_g \big ]^{-1} \bar P_\rho W_g P_0 \phantom{\sum_\lambda^\lambda} \notag \\
& = \sum_{\lambda=1,2} \int_{\mathbb{R}^3} P_0 \tilde w(r,k,\lambda) \big [ H_r + \frac{1}{2M} ( P - k )^2 + |k| - E_g \big ]^{-1} w(r,k,\lambda) P_0 \d k \notag \\
& \quad - \sum_{\lambda=1,2} \int_{|k| \le \rho} P_0 \tilde w(r,k,\lambda) \big [ e_0 + \frac{1}{2M} ( P - k )^2 + |k| - E_g \big ]^{-1} \notag \\
& \quad \phantom{- \sum_{\lambda=1,2} \int_{|k| \le \rho} P_0 \bar w(r,k,\lambda) } \times ( \mathds{1} \otimes P_{\phi_0} \otimes \mathds{1} ) w(r,k,\lambda) P_0 \d k  + O( |g|^{3} ). \label{eq:order_g2_1}
\end{align}
The second term in the right-hand-side of \eqref{eq:order_g2_1} is estimated as follows:
\begin{align}
& \bigg \| \sum_{\lambda=1,2} \int_{|k| \le \rho} P_0 \tilde w(r,k,\lambda) \big [ e_0 + \frac{1}{2M} ( P - k )^2 + |k| - E_g \big ]^{-1} \notag \\
&\phantom{ \sum_{\lambda=1,2} \int_{|k| \le \rho} } \times ( \mathds{1} \otimes P_{\phi_0} \otimes \mathds{1} ) w(r,k,\lambda) P_0 \d k \bigg \|  \le \sum_{\lambda=1,2} \int_{|k| \le \rho} \frac{ \mathrm{C} }{ |k|^2 } \d k \le \mathrm{C}' \rho.
\end{align}
Hence \eqref{eq:order_g2} is proven.
\end{proof}

\bibliographystyle{amsalpha}

\end{document}